\documentclass[pra,aps,preprintnumbers,superscriptaddress,nofootinbib]{revtex4}
\usepackage[utf8]{inputenc}

\usepackage{mathrsfs}
\usepackage{bbm}
\usepackage{palatino}
\usepackage{amsmath}
\usepackage{amssymb}
\usepackage{amsthm}  
\usepackage{ifthen}
\usepackage{stmaryrd}
\usepackage[colorlinks=true,linkcolor=black,citecolor=black,plainpages=false,pdfpagelabels]{hyperref}
\usepackage{tikz}
\usepackage{verbatim}
\usetikzlibrary{shapes.geometric,plotmarks,backgrounds,fit}
\usepackage{thmtools}

\newtheorem{thm}{Theorem}
\newtheorem{lem}[thm]{Lemma}
\newtheorem{cor}[thm]{Corollary}

\newtheorem{defin}[thm]{Definition}
\newtheorem{claim}[thm]{Claim}

\newcommand{\sket}[1]{{\ensuremath{\lvert#1\rangle}}}
\newcommand{\lket}[1]{{\ensuremath{\left\lvert#1\right\rangle}}}
\newcommand{\ket}[1]{\mathchoice{\lket{#1}}{\sket{#1}}{\sket{#1}}{\sket{#1}}}

\newcommand{\sbra}[1]{{\ensuremath{\langle#1\rvert}}}
\newcommand{\lbra}[1]{{\ensuremath{\left\langle#1\right\rvert}}}
\newcommand{\bra}[1]{\mathchoice{\lbra{#1}}{\sbra{#1}}{\sbra{#1}}{\sbra{#1}}}

\newcommand{\sbraket}[2]{{\ensuremath{\langle#1\rvert#2\rangle}}}
\newcommand{\lbraket}[2]{{\ensuremath{\left\langle#1\!\left\rvert\vphantom{#1}#2\right.\!\right\rangle}}}
\newcommand{\braket}[2]{\mathchoice{\lbraket{#1}{#2}}{\sbraket{#1}{#2}}{\sbraket{#1}{#2}}{\sbraket{#1}{#2}}}

\newcommand{\sketbra}[2]{{\ensuremath{\lvert #1\rangle\langle #2\rvert}}}
\newcommand{\lketbra}[2]{{\ensuremath{\left\lvert #1\middle\rangle\middle\langle #2\right\rvert}}}
\newcommand{\ketbra}[2]{\mathchoice{\lketbra{#1}{#2}}{\sketbra{#1}{#2}}{\sketbra{#1}{#2}}{\sketbra{#1}{#2}}}

\newcommand{\proj}[1]{\ketbra{#1}{#1}}

\newcommand{\floor}[1]{\left\lfloor #1 \right\rfloor}

\newcommand{\ident}{\id}
\DeclareMathOperator{\tr}{\mathrm{Tr}}

\DeclareMathOperator{\Herm}{Herm}

\newcommand{\demi}{\frac{1}{2}}
\newcommand{\hmax}{{\rm H}_{\max}}
\newcommand{\hmin}{{\rm H}_{\min}}
\newcommand{\htwo}{{\rm H}_{2}}
\newcommand{\hmineps}{{\rm H}_{\min}^{\varepsilon}}

\newcommand{\mbC}{\mathbb{C}}

\newcommand{\mbE}{\mathbb{E}}

\newcommand{\mbR}{\mathbb{R}}

\newcommand{\mfS}{\mathfrak{S}}

\renewcommand{\otimes}{\varotimes}

\DeclareMathOperator{\opt}{opt}
\newcommand{\commentout}[1]{}
\DeclareMathOperator{\Pos}{Pos}

\def\01{\{0,1\}}

\newcommand{\nc}{\newcommand}
\nc{\dg}{\dagger}
\nc{\cA}{{\cal A}}
\nc{\cB}{{\cal B}}
\nc{\cC}{{\cal C}}
\nc{\cD}{{\cal D}}
\nc{\cE}{{\cal E}}
\nc{\cF}{{\cal F}}
\nc{\cG}{{\cal G}}
\nc{\cH}{{\cal H}}
\nc{\cI}{{\cal I}}
\nc{\cJ}{{\cal J}}
\nc{\cK}{{\cal K}}
\nc{\cL}{{\cal L}}
\nc{\cM}{{\cal M}}
\nc{\cN}{{\cal N}}
\nc{\cO}{{\cal O}}
\nc{\cP}{{\cal P}}
\nc{\cR}{{\cal R}}
\nc{\cS}{{\cal S}}
\nc{\cT}{{\cal T}}
\nc{\cU}{{\cal U}}
\nc{\cV}{{\cal V}}
\nc{\cX}{{\cal X}}
\nc{\cZ}{{\cal Z}}

\nc{\eps}{\varepsilon}
\nc{\e}{\varepsilon}

\nc{\binent}{h}

\nc{\id}{{\operatorname{id}}}

\nc{\entI}{{\bf I}}
\nc{\entIarrow}{{\bf I}^{\leftarrow}}
\nc{\entH}{{\bf H}}
\nc{\entS}{{\bf S}}
\nc{\entHmin}{\hmin} %\mathbf{H}_{\min}}
\nc{\entHtwo}{\htwo}%{\mathbf{H}_{2}}
\nc{\entHmax}{\hmax}%\mathbf{H}_{\max}}
\nc{\aentHmin}{\hat{\mathbf{H}}_{\min}}

\nc{\hin}{\mathcal{H}_{\rm in}}
\nc{\hout}{\mathcal{H}_{\rm out}}
\nc{\msg}{M}

\nc{\qubitchannel}{\id_2}
\nc{\bitchannel}{\overline{\id}_2}
\nc{\clchannel}{\overline{\id}}

\nc{\ox}{\otimes}

\nc{\sym}{ { \rm sym } }

\nc{\states}{\cS}

\newcommand{\rhat}{\widehat{\rho}}

\nc{\supp}{\mathrm{supp}}

\newcommand{\eqdef}	{\stackrel{\textrm{def}}{=}}

\newcommand{\prob}[1]	{\mathbf{Pr}\left\{ #1 \right\}}
\newcommand{\pr}[1]	{\prob{#1}}

\begin{document}
\title{Entanglement sampling and applications}
\author{Frédéric \surname{Dupuis}}
\email[]{dupuis@cs.au.dk}
\affiliation{Department of Computer Science, Aarhus University, Åbogade 34, 8200 Aarhus, Denmark}
\affiliation{Institute for Theoretical Physics, ETH Z\"{u}rich, 8093 Zürich, Switzerland}
\author{Omar \surname{Fawzi}}
\email[]{ofawzi@phys.ethz.ch}
\affiliation{Institute for Theoretical Physics, ETH Z\"{u}rich, 8093 Zürich, Switzerland}
\author{Stephanie \surname{Wehner}}
\email[]{steph@locc.la}
\affiliation{Centre for Quantum Technologies, National University of Singapore, 2 Science Drive 3, 117543 Singapore}
\affiliation{School of Computing, National University of Singapore, 13 Computing Drive, 117417 Singapore}

\date{\today}
\begin{abstract}
	A natural measure for the amount of quantum information that a physical system $E$ holds about another system $A = A_1,\ldots,A_n$ is given by the min-entropy $\hmin(A|E)$.
	Specifically, the min-entropy measures the amount of entanglement between $E$ and $A$, and is the relevant measure when analyzing a wide variety of problems 
	ranging from randomness extraction in quantum cryptography, decoupling used in channel coding, to physical processes such as thermalization or the thermodynamic
	work cost (or gain) of erasing a quantum system.
	As such, it is a central question to determine the behaviour of the min-entropy after some process $\mathcal{M}$ is applied to the system $A$. Here we introduce a new generic tool relating the resulting min-entropy to the original one, and apply it to several settings of interest.

	\begin{itemize}
		\item A simple example of such a process is the one of \emph{sampling}, where a subset $S$ of the systems $A_1,\ldots,A_n$ is selected at random. The question is then to quantify the entanglement that $E$ has with the selected systems $A_S$, i.e., $\hmin(A_S|ES)$ as a function of the original $\hmin(A|E)$. 
			This has two applications by itself. First, it directly provides the first local quantum-to-classical randomness extractors for use in quantum cryptography, as well as decoupling operations acting on only a small fraction $A_S$ of the input $A$. Moreover, it gives lower bounds on the dimension of $k$-out-of-$n$ fully quantum random access encodings.
			%\omar{why separate these txwo applications, isn't the first a special case of the second?} \steph{That's true, but I'm not sure the non expert realizes this. I mean we also sent this to crypto where we had the QC extractor thing and of those people (who may be interested in the long version if our paper is hopefully accepted) probably noone knows what decoupling is. Conversely, people who know what decoupling is may not know this..}

		\item Another natural example of such a process is a measurement in e.g., BB84 bases commonly used in quantum cryptography. We establish the first entropic uncertainty relations with quantum side information that are nontrivial whenever $E$ is not maximally entangled with $A$. 

		\item As a consequence, we are able to prove optimality of quantum cryptographic schemes in the noisy-storage model (NSM).
			This model allows for the secure implementation of two-party cryptographic primitives under the assumption that the adversary cannot store quantum information perfectly. A special case is the bounded-quantum-storage model (BQSM) which assumes that the adversary's quantum memory device is noise-free but limited in size. Ever since the inception of the BQSM \cite{DFSS05}, it has been a vexing open question to determine whether security is possible as long as the adversary can only store strictly less than 
the number of qubits $n$ transmitted during the protocol.  
Here, we show that security is even possible as long as the adversary's device is not larger than $n - O(\log^2 n)$ qubits, which finally settles the fundamental limits of the BQSM.

	\end{itemize}
\end{abstract}
\maketitle
%\tableofcontents

\section{Introduction}
A central task in quantum theory is to effectively quantify the amount of information that some system $E$ holds about some classical or quantum data $A$.
For classical data, i.e., $A$ is a string $X^n = X_1,\ldots,X_n$, 
the \emph{min-entropy} $\hmin(X^n|E)$ forms a particularly
relevant measure because it determines the length of a secure key that can be 
obtained from $X^n$. This is the setting typically considered in quantum key distribution where $E$ is some information that an adversary Eve has gathered during the course
of the protocol, and $X^n$ is the so-called raw key.
More precisely, the maximum number $\ell$ of 
(almost) random bits~\footnote{We restrict ourselves to bits in the introduction, however, all our results also apply to higher 
dimensional alphabets.} that can be obtained from $X^n$ that are both uniform and uncorrelated from $E$ 
obeys $\ell \approx \hmin(X^n|E)$, if $E$ is classical~\cite{ILL89} and quantum~\cite{Ren05}. The process
by which such randomness is obtained is known as \emph{randomness extraction} (see~\cite{vadhan:survey} for a survey) or privacy amplification. 
Classically, a (strong) 
randomness extractor is simply a set of functions $\mathcal{F} = \{f: \{0,1\}^n \rightarrow \{0,1\}^\ell\}$ such that 
for almost all functions $f \in \mathcal{F}$, its output $f(X^n)$ is close to uniform 
and uncorrelated from the adversary, even if he learns which
function was applied. That is, the output is of the form $\rho_{F(X)EF} \approx
\id/2^n \otimes \rho_{EF}$. A well known example of such a set $\mathcal{F}$ is a set of two-universal hash functions which 
are used in quantum cryptography to turn a raw key $X^n$ into a secure key $f(X^n)$. 
The min-entropy also has a very intuitive interpretation as it can be expressed as $\hmin(X^n|E) = - \log P_{\rm guess}(X^n|E)$ where $P_{\rm guess}(X^n|E)$ is the probability that the adversary manages to guess $X^n$ maximized over all measurements on $E$~\cite{KRS09}. 

What can we say in the case of quantum data $A$? 
It turns out that the fully quantum min-entropy $\hmin(A|E)$ provides us with a similarly useful way to quantify the amount of information that $E$ holds about $A$. 
Its first significance is to quantum cryptography where $E$ is again held by an adversary.
More specifically, it has been shown that a quantum-to-classical extractor (QC-extractor) can produce exactly $\ell \approx \hmin(A|E)+\log |A|$ classical bits which are uniform and uncorrelated from $E$~\cite{BFW12}. Instead of applying functions to a classical string, a QC-extractor consists of a set of projective measurements on $A$ giving a classical string as a measurement outcome. Such extractors form a useful 
tool in two-party quantum cryptography where one might have an estimate of $\hmin(A|E)$, but not of the min-entropy of any classical 
string $X^n$ produced from $A$. Thus $\hmin(A|E)$ is directly related to the amount of cryptographic randomness that can be produced from $A$. 

More generally, the min-entropy is of significance in quantum information theory where it quantifies the number of qubits of $A$ that can be decoupled from $E$~\cite{HHYW08, DBWR10}.
A decoupling operation is given by a quantum operation
$\mathcal{K}_{A \rightarrow B}$ on the system $A$ that (approximately) transforms a state $\rho_{AE}$ to 
$\tau_B \otimes \rho_E$, where $\tau_B$ depends only on $\mathcal{K}$ but not on $\rho_A$. 
When $\tau_B = \id/|B|$ is the maximally mixed state, the operation $\mathcal{K}_{A\rightarrow B}$ again generates randomness with respect to $E$ and can hence be understood as a fully quantum-to-quantum extractor (QQ-extractor). 
When decoupling is used in quantum information theory, $E$ is typically the environment of a channel $\mathcal{N}_{\bar{A}\rightarrow B}$ acting on half of a maximally entangled state $\Phi_{A\bar{A}}$, and the 
number of qubits that can be decoupled relates directly to the number of qubits that can be 
transmitted correctly 
through the channel $\mathcal{N}_{\bar{A}\rightarrow B}$ (see~\cite{Dup09} for an in-depth exposition).
Recently, the min-entropy has also gained prominence in related areas such as the study of thermalization~\cite{adrian:thesis,lidia:inprep} and well as the thermodynamics work cost (or gain!) of erasing a quantum system~\cite{renato:workGain}. 

%However, the min-entropy is central even for the more general 
%quantum-to-quantum extractor (QQ-extractor) which yields a quantum output $Q$ by performing an operation on $A$. 
%More specifically, a QQ-extractor~\footnote{For experts in quantum information theory, QQ-extractors in turn 
%form a special class of decoupling operations.} can output a fully random quantum state (fully mixed state) $\rho_Q = \id/|Q|$ that is uncorrelated from $E$ of dimension at most $|Q| \approx \hmin(A|E)$~\cite{decoupling paper}. Clearly, any QQ-extractor also gives a QC-extractor when followed by a measuremeng in the standard basis.

It turns out that the fully quantum min-entropy also enjoys a very appealing operational interpretation~\cite{KRS09}.
More precisely, 
\begin{align}\label{eq:minDef}
	\hmin(A|E) = - \log |A| \max_{\Lambda_{E\rightarrow \bar{A}}} F(\Phi_{A\bar{A}}^N,\id_{A} \otimes \Lambda_{E\rightarrow \bar{A}}(\rho_{AE}))^2\ ,
\end{align}
where $F$ is the fidelity (see below) and $\Phi_{A\bar{A}}^N$ is the normalized maximally entangled state across $A$ and $\bar{A}$. That is, $\hmin(A|E)$ measures how close $\rho_{AE}$ can be brought to the maximally entangled state by performing 
a quantum operation on $E$. Intuitively, this quantifies how close the adversary $E$ can bring himself to being quantumly maximally correlated with $A$ --- exactly analogous to maximizing his classical correlations by trying to guess $X^n$. 

\subsection{Results}

Given the significance of the min-entropy in quantum information, it is a natural question to ask how the min-entropy changes if we apply a quantum operation $\mathcal{M}$ to $A$. More precisely, one might ask how $\hmin(\mathcal{M}(A)|E)$ relates to $\hmin(A|E)$, for some completely positive trace preserving map $\mathcal{M}$. At present, we know that the min-entropy satisfies $\hmin(\mathcal{M}(A)|E) \geqslant \hmin(A|E)$ if $\mathcal{M}$ is unital~\cite{Tom12}. Can we make more refined statements?

Of particular interest to us is the case where the quantum system consist of $n$ qudits $A^n = A_1,\ldots,A_n$.
Our main result is to establish the following very general theorem for maps $\mathcal{M}$ with the property that we can diagonalize
$((\mathcal{M}^\dagger \circ \mathcal{M}) \otimes \id_{\bar{A}^n})(\Phi_{A^n\bar{A}^n}) = \sum_{s \in \{0,\ldots,d^2-1\}} \lambda_s \Phi_s$ where $A^n = A_1,\ldots,A_n$, $d=|A_j|$ is the
dimension of one of the individual qudits, $\Phi_{A^n\bar{A}^n}$ is again the unnormalized maximally entangled state, and 
$\{\Phi_s\}_s$ is a basis for the space $A^n \otimes \bar{A}^n$ 
consisting of maximally entangled vectors (see Sections~\ref{sec:prelim} and Section~\ref{sec:evol} for precise definitions and statement of the theorem). 
%Note that $\lambda_s$ are the eigenvalues of 
In other words, the unnormalized state $((\mathcal{M}^\dagger \circ \mathcal{M}) \otimes \id_{\bar{A}^n})(\Phi_{A^n\bar{A}^n})$ on $A^n \bar{A}^n$ has eigenvalues $\lambda_s$ and eigenvectors $\ket{\Phi_s}$.
In terms of the smooth min-entropy $\hmineps$, which, loosely speaking, is equal to the min-entropy except with error probability $\varepsilon$, our first contribution can be stated as

\begin{itemize}
	\item {\bf Main result (Informal)} For any partition of $\{0,\ldots,d^2-1\}^{n} = \mathfrak{S}_+ \cup \mathfrak{S}_-$ into subsets $\mathfrak{S}_+$,$\mathfrak{S}_-$ we have 
		$2^{-\hmin^{\varepsilon}(\mathcal{M}(A^n)|E)} \lessapprox \sum_{s \in \mathfrak{S}_+} \lambda_s 2^{-\hmin(A^n|E)} + (\max_{s \in \mathfrak{S}_-} \lambda_s) d^n$. 
\end{itemize}

At first glance, our condition on the maps $\mathcal{M}$ may seem rather unintuitive and indeed restrictive. Yet, it turns out that many interesting maps do indeed satisfy these conditions, allowing us to establish the following results.

\bigskip
\noindent
{\bf Entanglement sampling} 
In the study of classical extractors, a goal was to construct families of functions $f$ that are \emph{locally
computable}~\cite{Vad03}.
That is, if our goal were to extract only a very small number of key bits from a long string $X^n$ of length $n$, one might wonder whether this can be done efficiently in the sense that the functions $f$ depend only on a small number of bits of $X^n$. 
Classically, a very beautiful method to answer this question is to show that the min-entropy can in fact be \emph{sampled}~\cite{Vad03,NZ96}.
That is, 
if we choose a subset $S$ of the bits at random, then the min-entropy of the bits $X_S$ in that subset $S$ obeys
\begin{align}\label{eq:classicalSampling}
	\hmin(X_S|ES) \gtrapprox |S| R(\hmin(X^n|E)/n)\ ,
\end{align}
for some function $R$. The function $R$ can be understood
as a rate function that determines the relation of the original min-entropy rate $\frac{\hmin(X^n|E)}{n}$ to the min-entropy rate on a subset $S$ of the bits.
%where Figure~\ref{fig:sampling-plots} compares known values of the rate $R$ when $E$ is classical~\cite{Vad03} or quantum~\cite{Wul10}.
In other words, min-entropy sampling says that
if $X^n$ is hard to guess, then even given 
the choice of subset $S$ it is tricky for the adversary to guess $X_S$. To see why this yields the desired functions $f$ note that one way to construct a randomness extractor would
be to first pick a random subset $S$, and then apply an arbitrary extractor to the much shorter bit string $X_S$. In the classical literature,
this is known as the sample-then-extract approach~\cite{Vad03}. 

Inspired by the classical results of Vadhan~\cite{Vad03}, it is a natural question whether there exists
QC-extractors which are efficient in the sense that the measurements $M \in \mathcal{M}$ only act on a small number of qubits of $A^n = A_1,\ldots,A_n$. 
Or, even more generally, whether there exist decoupling operations which depend on only very few qubits. 
As before, one way to answer this question in generality is to show that even the fully quantum min-entropy can be sampled - that is, that \emph{entanglement} can be sampled. 

\begin{itemize}
\item {\bf Entanglement sampling (Informal)} Entanglement sampling is possible for any quantum state $\rho_{A^nE}$, i.e., 
$\hmin^{\varepsilon}(A_S|ES) \gtrapprox |S|R(\hmin(A^n|E)/n)$ 
for the rate function $R$ plotted in Figure \ref{fig:sampling-plots}. See Theorem \ref{thm:h2-sampling} for a precise statement.
\end{itemize}

It should be noted that even the case of standard min-entropy sampling of a classical string $X^n$, 
but quantum side information $E$ has proved challenging. The results of~\cite{BARdW08} imply that sampling of classical strings is possible when the distribution over the strings $X^n$ is 
uniform (i.e., $\rho_{X^nE} = (1/2^n) \sum_{x \in \01^n} \proj{x} \otimes \rho_E^x$),
and the size of $E$ is bounded, and~\cite{KR07} has shown that sampling of blocks (but not individual bits) is possible. This was later refined in~\cite{Wul10} to show that bitwise sampling is also possible (see Figure~\ref{fig:sampling-plots} for a comparison of the rate function). Very roughly, the techniques used in~\cite{Wul10} relate the adversary's ability to guess the string $X^n$ to his ability to guess the XOR of bits in the string. Clearly, in the case of fully quantum $A^n$ such techniques cannot be used as it is indeed unclear what the XOR of qubits even means. 

As this is the first result on entanglement sampling, 
it required entirely novel techniques. More precisely, it inspired the even more general theorem sketched above, from which entanglement sampling follows by choosing an
appropriate map $\mathcal{M}$. As a byproduct, using the same techniques, we also obtain a stronger statement of sampling a classical string $X^n$ with respect to a quantum system $E$ in the sense that the rate $R$ is improved (see Figure~\ref{fig:sampling-plots} for a comparison).
What's more, we are able to show an even more precise statement in terms of the entropy $\htwo(A^n|E)_{\rho}$ - without any $\varepsilon$ error terms. Classically, this quantity is known as the (conditional) collision entropy. 
In general, it is very closely related to the min-entropy, and in fact enjoys a very similar 
operational interpretation. More specifically, it can be expressed in the same form as~\eqref{eq:minDef} where the optimization over all quantum operations $\Lambda_{E \rightarrow \bar{A}^n}$ is replaced by the so-called \emph{pretty good recovery} map $\Lambda^{\rm pg}_{E \rightarrow \bar{A}^n}$ which is close to optimal~\cite{BarnKnil02}. 

\bigskip
\noindent
{\bf Application to quantum random access codes }
Another application of our entanglement sampling result is to the fully quantum random access codes. Previous works have considered encodings of $n$ classical bits $X^n=X_1,\ldots,X_n$ into quantum states $\rho_E^{X^n}$ such that any desired bit can be retrieved with a particular success probability $p$~\cite{antv02,Nayak99}. 
This was later generalized to retrieving any subset of $k$ bits from the encoding~\cite{BARdW08}. The goal of~\cite{Nayak99,BARdW08} was to derive a bound on the necessary
dimension of $\rho_E^{X^n}$ as a function of $p$ when the string $X^n$ was chosen uniformly at random. Here, we prove dimension bounds for encoding $n$ \emph{qubits} $A^n=A_1,\ldots,A_n$ 
when we desire to recover any subset of $k$ qubits with a particular fidelity. Or, read in the opposite direction, we establish a bound on the fidelity as a function of the dimension (see Section~\ref{sec:qrac}).

\bigskip
\noindent
{\bf Uncertainty relations } Another consequence of our main result is a new uncertainty relation with quantum side information~\cite{BCCRR10} for measurements of $n$ qubits $A^n = A_1,\ldots,A_n$ in randomly chosen BB84~\cite{BB84} bases. 
Apart from the foundational consequences, such relations have found use in verifying the presence of entanglement~\cite{ur:experiment} as well as in quantum cryptography (see e.g.,~\cite{BFW12}).
Our result establishes the first entropic uncertainty relation with quantum side-information that uses a high-order entropy like the min-entropy and that is nontrivial as soon as the system being measured is not maximally entangled with the observer $E$. In other words, this shows a quantitative bound on the probability of successfully guessing the measurement outcome that is nontrivial as soon as $\hmin(A^n|E) > -n$.~\footnote{The fully quantum min-entropy can be negative up to $\hmin(A^n|E)=-n$ if $\rho_{A^nE}$ is the maximally entangled state.} 
\begin{itemize}
\item {\bf High-order entropic uncertainty relation for BB84 bases} If $X^n$ is obtained by measuring the system $A^n$ in a random BB84 bases $\Theta^n$, we have $\hmin(X^n|E \Theta^n) \geqslant n \cdot \frac{1}{2} \gamma\left(\frac{\hmin(A^n|E)}{n}\right)$, where the function $\gamma$ is plotted in Figure \ref{fig:uncertainty-plots}. See Theorem \ref{thm:ur-h2-bb84} and Corollary \ref{cor:ur-hmin-bb84} for precise statements.
\end{itemize}
We also prove uncertainty relations for qudit-wise measurements in mutually unbiased bases in Theorem \ref{thm:ur-h2-mub}.
Again, these results follow from our very general theorem sketched above, this time for a map $\mathcal{M}$ that represents randomly chosen measurements.

\bigskip
\noindent
{\bf Applications to the noisy-storage model } Our new uncertainty relations have several interesting applications to cryptography. The goal of two-party cryptography is to enable Alice and Bob to solve tasks in cooperation even if they do not trust each other. A classic example of such tasks are bit commitment and oblivious transfer.
Unfortunately, it has been shown that even using quantum communication, none of these tasks can be implemented 
securely without making assumptions~\cite{mayers:bitcom, lo:promise, lo:insecurity, lo&chau:bitcom,kretch:bc,bcs12}. 
What makes such tasks more difficult than quantum key distribution is that Alice and Bob cannot collaborate 
to check on any eavesdropper.
Instead, each party has to fend for itself. 

Nevertheless, because two-party computation is such a central part of modern cryptography, one is willing to make \emph{assumptions} on how powerful an attacker
can be in order to implement them securely. 
%Classically, such assumptions generally take the form of so-called computational assumptions, which are actually divided into two separate ones. First, it is assumed that a particular problem is difficult to solve in the sense that the attacker would need a lot of computation time to answer them.  Second, it is then assumed that the attackers computational resources are in fact limited, ensuring that he does not have enough time - namely not enough to solve said problem. 
Classically, such assumptions generally take the form of computational assumptions, where we assume that a particular mathematical problem cannot be solved in polynomial time.  Here, we consider \emph{physical} assumptions that can enable us to solve such tasks. 
In particular, can the sole assumption of a limited storage device lead to security~\cite{Maurer92b}? This is indeed the case and 
it was shown that security can be obtained if the attacker's \emph{classical} storage is limited~\cite{Maurer92b,cachin:bounded}. 
Yet, apart from the fact that
classical storage is cheap and plentiful, assuming a limited classical storage has one rather crucial caveat: If the honest players need
to store $n$ classical bits to execute the protocol in the first place, \emph{any} classical protocol can be broken if the attacker can store
more than roughly $n^2$ bits~\cite{maurer:imposs}. 
Motivated by this unsatisfactory gap, it was thus suggested to assume that the attacker's \emph{quantum} storage 
was bounded~\cite{BB84,DFRSS07,DFSS05,chris:id1,chris:id2}, or, more generally, noisy~\cite{WST08, STW08,KWW09}. 
The central assumption of the noisy-storage model is that during waiting times $\Delta t$ introduced in the protocol, the attacker can keep quantum information
only in his noisy quantum storage device; otherwise he is all-powerful (see Section~\ref{sec:nsm}). 

The assumption of bounded or noisy quantum storage offers significant advantages in that the proposed 
protocols do not require any quantum storage at all to be implemented by the honest parties. They are typically based on BB84~~\cite{KWW09} or six-state~\cite{BFW12} encodings, and indeed the first implementation of a bit commitment protocol has recently been performed experimentally~\cite{noisy:experiment}. So far it was known that there exist protocols that send $n$ qubits encoded in either the BB84 or six-state encoding, and that are secure as long as the adversary can only store strictly less than $n/2$ or $2n/3$ noise-free qubits respectively. 

Using our new techniques, we are able to show security of the primitive called \emph{weak string erasure}~\cite{KWW09} (see Section~\ref{sec:nsm}), 
which in turn can be supplemented with additional classical or quantum communication~\cite{yao} to obtain primitives such as bit commitment.

\begin{itemize}
\item {\bf Application 1: Bounded storage} There exists a weak string erasure protocol transmitting $n$ qubits that is secure as long as the adversary can store at most strictly less than $n - O(\log^2 n)$ qubits. The protocol does not require any quantum memory to be executed, and merely requires simple quantum operations and measurements. See Theorem \ref{thm:bqsm} for a precise statement.
\end{itemize}

It should be noted that no such protocol can be secure as soon as the adversary can store $n$ 
qubits, so our result is essentially optimal. Our result highlights the sharp contrast between 
the classical and the quantum bounded storage model and answers the main open question in the BQSM. 
The noisy-storage model offers an advantage over the case of bounded-storage not only for implementations using high-dimensional encodings such as the infinite-dimensional states sent in continuous variable experiments, 
but allows security even for arbitrarily large storage devices as long as the noise is large enough.~\footnote{
Note that if we have a general noisy storage device $\cF$ that can be simulated by an $m$-qubit noiseless channel, then security against a device that can reliably store at most $m$ qubits readily implies security against the device $\cF$. In fact, security in the noisy-storage model does follow rather directly from security in the bounded-storage model provided the so-called entanglement cost of the storage device is small enough~\cite[Lemma 18]{BBCW11}. The entanglement cost measures the number of noiseless channels (bounded storage) needed to simulate a certain number of noisy channels, in the presence of classical communication. However, the entanglement cost of a channel is in general larger than its quantum capacity. Thus, proving security in the noisy-storage model up to the quantum capacity does not follow directly from proving security in the bounded-storage model. 
% security in the bounded storage model implies rather directly security for a general channel up to some rate known as the entanglement cost.
%For information theory experts, we note that security in the noisy-storage model does not directly follow from security in the bounded-storage model because the entanglement cost of a channel~\cite{BBCW11} is not equal to its quantum capacity. The entanglement cost measures the number of noise free channels (bounded storage) needed to simulate a certain number of noisy channels, in the presence of classical communication. As such, it allows properties of noisy channels to be derived from properties of noiseless channels: a certain amount of noisy channels cannot be used to accomplish a certain task (like sending a certain number of qubits with a given fidelity) because otherwise also some number of noise free channels could be used by simulating the noisy ones.
}
Essentially, the noisy-storage model captures our intuition that security should be linked to how much information the adversary can store in his quantum memory. The first proofs linked security to the classical capacity~\cite{KWW09}, the entanglement cost~\cite{BBCW11} and finally the quantum capacity~\cite{BFW12}. The latter result used a protocol based on six-state encodings. 

\begin{itemize}
\item {\bf Application 2: Noisy storage} We 
significantly push the boundaries regarding when security is possible in the noisy-storage model (see Section~\ref{sec:nsm}). 
Furthermore, we link security of a BB84-based protocol to the quantum capacity of the adversary's storage device for the first time. See Theorem \ref{thm:nsm} for a precise statement.
\end{itemize}
%We expect that our results regarding the noisy-storage model are also essentially optimal \steph{hm}. 

\section{Preliminaries}\label{sec:prelim}

\subsection{Basic concepts and notation}

In quantum mechanics, a system such as Alice's or Bob's labs are described mathematically by \emph{Hilbert spaces}, denoted by $A, B, C,\ldots$. Here, we follow the usual convention in quantum cryptography and assume that all Hilbert spaces are finite-dimensional. We write $|A|$ for the dimension of $A$. A system of $n$ qudits is also denoted as $A^n = A_1,\ldots,A_n$, 
where we also use $|A|$ to denote the dimension of one single qudit in $A^n$.
The set of linear operators on $A$ is denoted by $\cL(\cA)$, and we write $\Herm(A)$ and $\Pos(A)$ for the set of hermitian and positive semidefinite operators on $A$ respectively. We denote the adjoint of an operator $M$ by $M^{\dagger}$. A \emph{quantum state} $\rho_{A}$ is an operator $\rho_{A}\in\states(A)$, where $\states(A)=\{\sigma_{A}\in\Pos(A) \mid \tr(\sigma_{A})=1\}$. We will often make use of \emph{operator inequalities}: whenever $X, Y \in \Herm(A)$, we write $X \leqslant Y$ to mean that $Y - X \in \Pos(A)$. 
A quantum operation is given by a completely positive map $\cM : \cL(A) \to \cL(C)$. A map $\cM$ is said to be completely positive if for any system $B$ and $X \in \Pos(A \otimes B)$ we have $(\cM \otimes \id)(X) \geqslant 0$ (see~\cite{hayashi:book} for properties of quantum channels).

Throughout, we use the shorthand $[d] = \{0, 1, \dots, d-1\}$. We will follow the convention to use $H$ to denote the unitary that takes the computational $\{\ket{0},\ket{1}\}$ to the Hadamard basis: $H\ket{0} = \frac{1}{\sqrt{2}}(\ket{0}+\ket{1}), H\ket{1} = \frac{1}{\sqrt{2}} (\ket{0}-\ket{1})$.
When considering $n$ qubits, we also use $H^{\theta^n} =  H^{\theta_1} \otimes \cdots \otimes H^{\theta_n}$ for the unitary defining the basis $\theta^n \in \01^n$.

\subsection{Entropies}

%In the following, let $\mathcal{S}_n$ denote the symmetric group on $n$ elements, and it will act in the obvious way on $A^n$. Furthermore, let $\mathcal{W}_d = \{ W_i : i \in 1,\dots,d^2 \}$ be the set of Weyl operators on a $d$-dimensional system $A$, with $W_1 = \ident$, and let $\ket{\Phi} = \sum_i \ket{i} \otimes \ket{i}$ for some orthonormal basis $\{ \ket{i} \}$. Also, for any string $s \in [ d^2 ]^{\times n}$, let $W_s = W_{s_1} \otimes \dots \otimes W_{s_n}$, and let $\Phi_s = (W_s \otimes \ident) \Phi_{A^n \bar{A}^n} (W_s^{\dagger} \otimes \ident)$. Note that $\left\{ \frac{1}{\sqrt{d}^n}\ket{\Phi_s} \right\}$ is an orthonormal basis for $A^n \otimes \bar{A}^n$.

Next to its operational interpretation given in~\eqref{eq:minDef}, the \emph{conditional min-entropy} of a positive operator $\rho_{AB}\in\cS(AB)$ can also be expressed as
\begin{align}
\label{eqn:def-hmin}
\entHmin(A|B)_{\rho}=\max_{\sigma_B \in \states(B)}\entHmin(A|B)_{\rho|\sigma} \ \text{ with } \
\entHmin(A|B)_{\rho|\sigma}=\max\left\{\lambda \in \mbR: 2^{-\lambda}\cdot\id_A \otimes \sigma_B \geqslant \rho_{AB}\right\}\ ,
\end{align}
where the symbol $\id_A$ refers to the identity on $A$. We use the subscript $\rho$ to emphasize the state $\rho_{AB}$ of which we evaluate the min-entropy.
The smoothed version is defined by
$\entHmin^{\varepsilon}(A|B)_{\rho}=\max_{\tilde{\rho}_{AB} \in \cB^{\varepsilon}(\rho_{AB})} \entHmin(A|B)_{\tilde{\rho}}\ ,$
where $\cB^{\varepsilon}(\rho)$ is the set of states at a distance at most $\varepsilon$ from $\rho$. We use the purified distance as the distance measure~\cite{tcr09}. 
We refer to~\cite{Tom12} for a review of the properties of the min-entropy.

It is simpler to state our results in terms of the related collision entropy defined for any $\rho_{AB} \in \Pos(A\otimes B)$ (possibly unnormalized) by
%\alternative{
\begin{align}
%\entHmin(A|B)_{\rho}=\max_{\sigma_B \in \states(B)}\entHmin(A|B)_{\rho|\sigma} 
%\ \text{ with } \
\label{eqn:def-h2}
\entHtwo(A|B)_{\rho} = -\log \left( \tr\left[ \left(\rho_B^{-1/4} \rho_{AB} \rho_B^{-1/4} \right)^2 \right] \right).
%\max\left\{\lambda \in \mbR: 2^{-\lambda} \cdot \id_A \otimes \rho_B \geqslant \rho_{AB}\right\}\ .
\end{align}
We also use the collision entropy conditioned on a general operator $\sigma_B \in \cS(B)$,
\begin{align}
\label{eqn:def-h2-rho-sigma}
\entHtwo(A|B)_{\rho|\sigma} = -\log \left( \tr\left[ \left(\sigma_{B}^{-1/4} \rho_{AB} \sigma_{B}^{-1/4} \right)^2 \right] \right).
\end{align}
Note that we \emph{do not} normalize this quantity by dividing by the trace of the operator $\rho_{AB}$. In fact the quantity $2^{-\entHtwo(A|B)_{\rho|\sigma}}$ gets multiplied by $\mu^2$ when $\rho$ is multiplied by $\mu \geq 0$.
%}

The two entropy measures $\hmin$ and $\htwo$ are closely related as shown in Lemmas \ref{lem:hmin-h2-fullyquantum}, \ref{lem:hmin-h2} and \ref{lem:smoothhmin-h2}. %What's more, a slight variant of $\htwo$ obeys $\htwo \approx \hmineps$, however we will focus on the exact result here. 
The collision entropy also has an appealing operational interpretation~\cite{bcw13} for normalized $\rho$ as
\begin{align}\label{eq:H2again}
	\entHtwo(A|B)_{\rho}= - \log \left( |A| F(\Phi_{AA'}^N,\id_A \otimes \Lambda_{E\rightarrow A'}^{\rm pg}(\rho_{AE}))^2 \right), 
\end{align}
where $F(\sigma_1,\sigma_2) = \tr\left(\sqrt{\sqrt{\sigma_1}\sigma_2\sqrt{\sigma_1}}\right)$ is the fidelity, and 
$\Lambda_{E \rightarrow A'}^{\rm pg}$ is the \emph{pretty good recovery} map~\cite{BarnKnil02} (see Section \ref{sec:oper-inter-h2} of the appendix).
Finally, we use the binary entropy function $\binent(x) = -x \log x - (1-x) \log(1-x)$.

For the curious reader, we note that the quantity~\eqref{eqn:def-h2-rho-sigma} has indeed also appeared in a slightly different guise in the context of 
norms employed for the study of mixing properties of quantum channels~\cite{temme:mixing}. Specifically, we have $\|\rho_{AB}\|^2_{2, \omega} = 2^{-\entHtwo(A|B)_{\rho|\sigma}}$
with $\omega = \id_{A} \otimes \sigma_B^{-1}$ for the norm $\|\cdot\|_{2,\omega}$ defined in~\cite{temme:mixing}.

\subsection{A convenient basis}

Throughout, we make use of a very convenient basis of maximally entangled states for the space $A \otimes \bar{A}$ where $\bar{A} \simeq A$.
The (unnormalized) maximally entangled state 
\begin{align}
	\ket{\Phi}_{A\bar{A}} = \sum_{a} \ket{a}_A \otimes \ket{a}_{\bar{A}}
\end{align}
will play an important role in our analysis. Here, the vectors $\ket{a}$ label the standard basis of $A$. We use $\ket{\Phi^N}_{A\bar{A}}$ to denote the normalized version $\ket{\Phi^N}_{A\bar{A}} = \frac{1}{\sqrt{|A|}} \ket{\Phi}_{A\bar{A}}$. We repeatedly use the following properties. For any operators $X$ and $Y$ acting on $A$, we have 
\begin{equation}
\label{eq:swap-trick-phi}
\tr[XY] = \tr[X \otimes \top(Y) \Phi_{A \bar{A}}]
\end{equation}
where $\top$ denotes the transpose map in the standard basis and $\Phi_{A\bar{A}} = \ket{\Phi} \bra{\Phi}_{A\bar{A}}$. Moreover, if $X : A \to C$ is a linear operator from $A$ to $C$ we have 
\begin{equation}
\label{eq:transpose-trick}
(X \otimes \id_{\bar{A}}) \ket{\Phi}_{A\bar{A}} = \sum_a X \ket{a}_A \otimes \ket{a}_{\bar{A}} = \sum_c  \ket{c}_C \otimes \sum_{a} \bra{c} X \ket{a} \ket{a}_{\bar{A}} = (\id_{C} \otimes \top(X)) \ket{\Phi}_{C\bar{C}}. 
\end{equation}
Again, the transpose $\top$ is taken with respect to the standard bases $\{\ket{a}\}$ and $\{\ket{c}\}$ of $A$ and $C$.
Using~\eqref{eq:swap-trick-phi}, one can naturally construct an orthogonal basis of $A \bar{A}$ by applying unitary transformations to $\ket{\Phi}$
that are orthogonal with respect to the Hilbert-Schmidt inner product. Define for $s \in [|A|^2]$,  $\ket{\Phi_s} = (W_s \otimes \id) \ket{\Phi}_{A\bar{A}}$
where $W_s$ denote the generalized Pauli operators (see e.g.,~\cite{BBRV02}), sometimes also called Weyl operators. In fact, all our results would hold for any unitary operators $W_s$ that are orthogonal with respect to the Hilbert-Schmidt inner product. We again use $\Phi_s = \proj{\Phi_s}$. 

In particular for $|A| = 2$, $W_0, W_1, W_2, W_3$ are the Pauli operators, and we obtain the well-known Bell basis
\begin{equation}\label{eqn:bellbasis01}
\Phi_0 = \left( \begin{array}{cccc}
1 & 0 & 0 & 1 \\
0 & 0 & 0 & 0 \\
0 & 0 & 0 & 0 \\
1 & 0 & 0 & 1
\end{array} \right), \quad 
\Phi_1 = \left( \begin{array}{cccc}
0 & 0 & 0 & 0 \\
0 & 1 & 1 & 0 \\
0 & 1 & 1 & 0 \\
0 & 0 & 0 & 0
\end{array} \right),
\end{equation}
\begin{equation}\label{eqn:bellbasis23}
\Phi_2 = \left( \begin{array}{cccc}
0 & 0 & 0 & 0 \\
0 & 1 & -1 & 0 \\
0 & -1 & 1 & 0 \\
0 & 0 & 0 & 0
\end{array} \right), \quad 
\Phi_3 = \left( \begin{array}{cccc}
1 & 0 & 0 & -1 \\
0 & 0 & 0 & 0 \\
0 & 0 & 0 & 0 \\
-1 & 0 & 0 & 1
\end{array} \right).
\end{equation}
Note that in this numbering scheme, $\Phi_2$ is the singlet.
%The classical analogue of $\Phi_{A\bar{A}}$ is defined as $T_{A\bar{A}} = \sum_{a} \proj{aa}_{A\bar{A}}$. If $P_i$ denotes the cyclic shift $P_i = \sum_{a} \ket{a+i\mod |A|} \bra{a}$, 
%we can define $T_s = (P_s \otimes \id) T (P_s^{\dagger} \otimes \id)$ for $s \in \{0, \dots, |A| - 1\}$. Note that $P_i$ is sometimes also referred to as the generalized Pauli $X_{|A|}$
%in dimension $|A|$.

For $n > 0$, we will denote by $A^n$ the system $\bigotimes_{i=1}^n A_i$, where each $A_i$ is a copy of $A$. Furthermore, if $S \subseteq \{1,\dots,n\}$, we write $A_S$ to denote $\bigotimes_{i \in S} A_i$.  In other words, $A^n$ consists of $n$ copies of the system $A$, and $A_S$ contains the copies that correspond to indices in $S$. In such a setting the dimension of the system $A$ is denoted $d$. We can naturally define for $s \in [d^2]^n$, $\ket{\Phi_s} = \otimes_{i=1}^n \ket{\Phi_{s_i}}_{A_i \bar{A}_i}$. We then have that $\{ \frac{1}{\sqrt{d^n}} \ket{\Phi_s} \}_s$ is an orthonormal basis of $A^n \bar{A}^n$. 
%Similarly we define $T_s \in \Pos(A^n \bar{A}^n)$ for $s \in [d]^n$ by $T_s = \otimes_{i=1}^n (T_{s_i})_{A_i \bar{A}_i}$. Note that $\sum_{s \in [d]^n} T_s = \id$. 
For such strings $s$, we denote $\supp(s) = \{i \in \{1, \dots, n\} : s_i \neq 0\}$ and $|s| = |\supp(s)|$.

%Given two systems of identical dimension $A$ and $A'$, we will define $\ket{\Phi}_{AA'}$ as  $\ket{\Phi}_{AA'} = \frac{1}{\sqrt{d_A}} \sum_i
%\ket{ii}_{AA'}$ and $T_{AA'} = \sum_i \ket{ii}\bra{ii}_{AA'}$.

%
%In the following, let $\mathcal{S}_n$ denote the symmetric group on $n$ elements, and it will act in the obvious way on $X^n$. Furthermore, let $P_i = \sum_j \ket{j+i}\bra{j}$ be a cyclic permutation on the basis vectors of one copy of $X$, and let $\ket{T} = \sum_i \ket{ii}\bra{ii}$. Also, for any string $s \in [ d ]^{\times n}$, let $P_s = P_{s_1} \otimes \dots \otimes P_{s_n}$, and let $T_s = (P_s \otimes \ident) T_{A^n \bar{A}^n} (P_s^{\dagger} \otimes \ident)$. Note that $\sum_{s \in [d]^{\times n}} T_s = \ident$.

\section{Evolution of $\entHtwo$ under general maps}
\label{sec:evolution-h2-general}\label{sec:evol}
In this section, we derive constraints on the evolution of the conditional collision entropy $\entHtwo$ when the system $A^n$ undergoes some transformation described by a completely positive map $\cM$. Our results on entanglement sampling and uncertainty relations are obtained by evaluating this bound for particular channels $\cM$. 
\begin{thm}
\label{thm:general-h2}
Let $\mathcal{M}_{A^n \rightarrow C}$ be a completely positive map such that $( (\mathcal{M}^{\dagger} \circ \mathcal{M})_{A^n} \otimes \id_{\bar{A}^n})(\Phi_{A^n\bar{A}^n}) = \sum_{s \in [d^2]^n} \lambda_s \Phi_s$ and let $\rho_{A^nE} \in \cS(A^nE)$ be a state, where $A^n=A_1,\ldots,A_n$ is comprised of $n$ qudits of dimension $d$. Then for any partition $[d^2]^n = \mfS_+ \cup \mfS_-$ into subsets $\mfS_+$ and $\mfS_-$, we have
%	\begin{equation}
%	\label{eqn:general-h2}
%	%\tr\left[\left(\rho^{-1/4}_E (\cM \otimes \id)(\rho_{A^nE}) \rho_E^{-1/4} \right)^2\right] 
%	2^{-\entHtwo(C|E)_{\cM(\rho)|\rho}}
%	\leqslant \sum_{s \in \mfS_+} \lambda_s 2^{-\entHtwo(A^n|E)_{\rho}} + (\max_{s \in \mfS_-} \lambda_s) d^n.
%	\end{equation}
	\begin{align}
	2^{-\entHtwo(C|E)_{\cM(\rho)|\rho}} &\eqdef \tr\left[\left(\rho^{-1/4}_E (\cM \otimes \id)(\rho_{A^nE}) \rho_E^{-1/4} \right)^2\right] \notag \\
	&\leqslant \sum_{s \in \mfS_+} \lambda_s 2^{-\entHtwo(A^n|E)_{\rho}} + (\max_{s \in \mfS_-} \lambda_s) d^n. \label{eqn:general-h2}
	\end{align}
\end{thm}
Note that in the case where $\cM$ is $\mu$-trace preserving, i.e., $\cM^{\dagger}(\id_{A}) = \mu \id_A$, we can directly relate $\entHtwo(C|E)_{\cM(\rho)|\rho}$ to $\entHtwo(C|E)_{\cM(\rho)}$. In fact, we have $\tr_A[\cM(\rho_{AE})] = \mu \rho_E$ and
	\[
	2^{-\entHtwo(C|E)_{\cM(\rho)}} = \tr\left[\left((\mu \rho)^{-1/4}_E (\cM \otimes \id)(\rho_{A^nE}) (\mu \rho_E)^{-1/4} \right)^2\right] = \frac{1}{\mu}2^{-\entHtwo(A|E)_{\cM(\rho)|\rho}}.
	\]
A statement for the smooth min-entropy follows directly by applying Lemma~\ref{lem:smoothhmin-h2}. In fact, we can first normalize the state $\cM(\rho)^N = \cM(\rho)/\tr[\cM(\rho)]$, then we obtain for $\rho \in \cS(AE)$,
\begin{align*}
2^{-\entHmin^{\e}(C|E)_{\cM(\rho)^N}} 
&\leq \frac{2}{\e^2} \cdot 2^{-\entHtwo(C|E)_{\cM(\rho)^N|\rho}} \\
&= \frac{2}{\e^2 \tr[\cM(\rho)]^2} \cdot 2^{-\entHtwo(C|E)_{\cM(\rho)|\rho}} \\
&\leq \frac{2}{\e^2 \tr[\cM(\rho)]^2}  \left( \sum_{s \in \mfS_+} \lambda_s 2^{-\entHtwo(A^n|E)_{\rho}} + (\max_{s \in \mfS_-} \lambda_s) d^n \right) \\
&\leq \frac{2}{\e^2 \tr[\cM(\rho)]^2} \left( \sum_{s \in \mfS_+} \lambda_s 2^{-\entHmin(A^n|E)_{\rho}} + (\max_{s \in \mfS_-} \lambda_s) d^n \right).
\end{align*}

The maps $\cM$ of interest typically have some symmetry. For example, if the map $\cM$ is invariant under permutations of the $n$ systems $A_1, \dots, A_n$, then the coefficients $\lambda_s$ only depend on the type of $s$, i.e., the number of times each symbol in $[d^2]$ occurs in $s$. In fact, all of the examples we consider here are such that $\lambda_s$ only depends on the weight $|s| = |\{i \in [n]: s_i \neq 0\}|$.
\begin{proof}
	Let $\tilde{\rho}_{A^nE} = \rho_E^{-1/4}\rho_{A^nE} \rho_E^{-1/4}$, and let $\rhat_{A^n \bar{A}^n} = \tr_{E \bar{E}}[(\tilde{\rho}_{A^nE} \otimes \top(\tilde{\rho}_{\bar{A}^n\bar{E}})) \Phi_{E\bar{E}}]$. Note that $\rhat_{A^n \bar{A}^n} \geq 0$ and $\tr[\rhat_{A^n \bar{A}^n}] = \tr[\tilde{\rho}_E^2] = 1$.
	Furthermore, define $\bar{\mathcal{M}}$ by $\bar{\cM}(X) = \top(\cM(\top(X)))$ for all $X$, so that $\bar{\mathcal{M}}(\top(X)) = \top(\mathcal{M}(X))$. Our first goal is to rewrite $\entHtwo(C|E)_{\cM(\rho)|\rho}$ in terms of the basis $\{\Phi_s\}_s$.
	We obtain from~\eqref{eq:swap-trick-phi}
	\begin{align*}
		2^{-\htwo(C|E)_{\cM(\rho)|\rho}} &= \tr[\mathcal{M}(\tilde{\rho}_{A^nE})^2]\\
		&= \tr[(\mathcal{M}(\tilde{\rho}_{A^nE}) \otimes \top(\mathcal{M}(\tilde{\rho}_{\bar{A}^n\bar{E}}))) \Phi_{C \bar{C}} \otimes \Phi_{E\bar{E}}]\\
		&= \tr[(\mathcal{M}(\tilde{\rho}_{A^nE}) \otimes \bar{\mathcal{M}}(\top(\tilde{\rho}_{\bar{A}^n\bar{E}}))) \Phi_{C \bar{C}} \otimes \Phi_{E\bar{E}}]\\
		&= \tr[(\tilde{\rho}_{A^nE} \otimes \top(\tilde{\rho}_{\bar{A}^n\bar{E}})) ( ({\mathcal{M}^{\dagger}}) \otimes ({\bar{\mathcal{M}}^{\dagger}}))(\Phi_{C\bar{C}}) \otimes \Phi_{E\bar{E}}].
\end{align*}
Now write the completely positive map $\cM$ in the Kraus representation. There exists linear operators $K_i : A \to C$ such that $\cM(X) = \sum_i K_i X K_i^{\dagger}$. 
Using \eqref{eq:transpose-trick}, we have
\begin{align*}
(\cM_{A^n\rightarrow C} \otimes \id_{\bar{A}^n})( \Phi_{A^n \bar{A}^n})
&= \sum_i (K_i \otimes \id_{\bar{A}^n}) \Phi_{A^n \bar{A}^n} (K_i^{\dagger} \otimes \id_{\bar{A}^n}) \\
&= \sum_i (\id_C \otimes \top(K_i)) \Phi_{C \bar{C}} (\id_C \otimes \top(K_i)^{\dagger}) \\
&= (\id_C \otimes {\bar{\mathcal{M}}^{\dagger}})(\Phi_{C\bar{C}}).
\end{align*} 
Thus, we obtain 
using the definition of $\rhat_{A^n \bar{A}^n}$ and the condition on $\mathcal{M}$
\begin{align}
	2^{-\htwo(C|E)_{\cM(\rho)|\rho}} &= \tr[(\tilde{\rho}_{A^nE} \otimes \top(\tilde{\rho}_{\bar{A}^n\bar{E}})) ( ({\mathcal{M}^{\dagger}} \circ {\mathcal{M}}) \otimes \id_{\bar{A}^n})(\Phi_{A^n\bar{A^n}}) \otimes \Phi_{E\bar{E}}] \notag \\
	&= \tr[\rhat_{A^n \bar{A}^n} ( ({\mathcal{M}^{\dagger}} \circ {\mathcal{M}}) \otimes \id_{\bar{A}^n})(\Phi_{A^n\bar{A^n}})] \notag \\
	&= \sum_{s \in [d^2]^n} \lambda_s \tr[\rhat_{A^n \bar{A}^n} \Phi_s]. \label{eqn:h2-phis}
\end{align}
We now prove the two key constraints on the terms $\tr[\rhat_{A^n \bar{A}^n} \Phi_s]$ we will be using. First, we have a global constraint. Note that the set of vectors $\{\frac{1}{\sqrt{d^n}} \ket{\Phi_s}\}_{s \in [d^2]^n}$ forms an \emph{orthonormal} basis and thus $\id_{A^n \bar{A}^n} = \frac{1}{d^n}\sum_{s \in [d^2]^n} \Phi_s$. This yields
\begin{align}
	\sum_{s \in [d^2]^n} \tr[\rhat_{A^n \bar{A}^n} \Phi_s] &= d^n\tr[\rhat_{A^n\bar{A}^n}] = d^n.
\label{eqn:total-constraint}
\end{align}

The second observation concerns the individual terms $\tr[\rhat_{A^n\bar{A}^n} \Phi_s]$. For any $s$,
\begin{align*}
	\tr[\rhat_{A^n\bar{A}^n} \Phi_s] &= \tr[\rhat_{A^n\bar{A}^n} (W_s \otimes \id_{\bar{A}^n}) \Phi_{A^n\bar{A}^n} (W_s^{\dagger} \otimes \id_{\bar{A}^n})] \\
	&= \tr[ \left(W_s^{\dagger} \tilde{\rho}_{A^nE} W_s \otimes \top(\tilde{\rho}_{\bar{A}^n\bar{E}}) \right) \Phi_{A^n \bar{A}^n} \otimes \Phi_{E\bar{E}}] \\
&= \tr[W_s^{\dagger} \tilde{\rho}_{A^nE} W_s \tilde{\rho}_{A^nE})] \\
&\leqslant \tr[\tilde{\rho}^2_{A^nE}] = 2^{-\entHtwo(A^n|E)_{\rho}},
\end{align*}
using the Cauchy-Schwarz inequality in the form $\tr[XY] \leq \sqrt{\tr[X^2] \tr[Y^2]}$ with $X = W_s^{\dagger} \tilde{\rho}_{A^nE} W_s$ and $Y=\tilde{\rho}_{A^nE}$. Also, observe that the positivity of $\rhat_{A^n \bar{A}^n}$ and $\Phi_s$ implies that $\tr[\rhat_{A^n \bar{A}^n} \Phi_s] \geqslant 0$. Thus, we have
\begin{equation}
\label{eqn:individual-constraint}
0 \leqslant \tr[\rhat_{A^n\bar{A}^n} \Phi_s] \leqslant 2^{-\entHtwo(A^n|E)_{\rho}}.
\end{equation}
Applying inequalities \eqref{eqn:total-constraint} and \eqref{eqn:individual-constraint} to \eqref{eqn:h2-phis}, we obtain the desired result. 
\end{proof}

We remark that equations \eqref{eqn:total-constraint} and \eqref{eqn:individual-constraint} are the only properties of the operator $\rho_{A^nE}$ that we are using. This means that the result would also hold for possible operators $\rho_{A^nE}$ that do not correspond to states but still satisfy these conditions. 
%\fred{Doesn't that just mean nonnormalized states? Or did you mean that it would also hold for operators $\rhat_{A^n \bar{A}^n}$ that don't correspond to states?}

\section{Applications}

We now derive several interesting consequences of Theorem~\ref{thm:general-h2}. All of these follow by making an appropriate choice for the map $\mathcal{M}$.

\subsection{Quantum-quantum min-entropy sampling}
\label{sec:ent-sampling}

\subsubsection{Statement}
We now state our results on entanglement sampling. 
The theorem below deals with the following scenario: we have $n$ qudits and we choose a subset of them of size $k$ uniformly at random. We have a lower bound on the collision entropy of the whole state conditioned on some quantum side-information $E$; the theorem then gives a lower bound on the conditional collision entropy of the sample. 
%Theorem \ref{thm:h2-sampling} treats the case when the systems being sampled are quantum, and Theorem \ref{thm:cq-h2-sampling} deals with the case when they are classical. 
The rate function obtained is plotted in Figure \ref{fig:sampling-plots} together with an upper bound on the optimal rate function given by a particular example presented in Theorem \ref{thm:upper-bound-rate-qq}. The same figure also shows plots of classical-quantum sampling results that are discussed in Section \ref{sec:classical-quantum}.
% together with previous min-entropy sampling results, as well as upper bounds on optimal rate functions given by particular examples presented in Theorem \ref{thm:upper-bound-rate-qq} for the fully quantum case, and Theorem \ref{thm:upper-bound-rate-cq} for the classical-quantum case.

%In the following, let $\mathcal{S}_n$ denote the symmetric group on $n$ elements, and it will act in the obvious way on $A^n$. Furthermore, let $\mathcal{W}_d = \{ W_i : i \in 1,\dots,d^2 \}$ be the set of Weyl operators on a $d$-dimensional system $A$, with $W_1 = \ident$, and let $\ket{\Phi} = \sum_i \ket{i} \otimes \ket{i}$ for some orthonormal basis $\{ \ket{i} \}$. Also, for any string $s \in [ d^2 ]^{\times n}$, let $W_s = W_{s_1} \otimes \dots \otimes W_{s_n}$, and let $\Phi_s = (W_s \otimes \ident) \Phi_{A^n \bar{A}^n} (W_s^{\dagger} \otimes \ident)$. Note that $\left\{ \frac{1}{\sqrt{d}^n}\ket{\Phi_s} \right\}$ is an orthonormal basis for $A^n \otimes \bar{A}^n$.

\begin{thm}%[name=Quantum min-entropy sampling, label=thm:h2-sampling, restate=RestateQQsampling]
\label{thm:h2-sampling}
	Let $\rho_{A^nE} \in \mathcal{S}(A^n E)$ and $1 \leqslant k\leqslant n$, let $d = |A|$ be the dimension of a single system, and 
	let $h_2 := \frac{\htwo(A^n|E)_{\rho}}{n}$. Then, we have for $n > d^2$
	\begin{equation}
	\label{eqn:h2-sampling}
	 2^{-\htwo(A_S|ES)_{\rho}} = \mbE_{S \subseteq [n], |S|=k} 2^{-\htwo(A_S|E)_{\rho}} \leqslant 2^{-kR_d(h_2) + \log (n^2+1)}, 
	 \end{equation}
	where $R_d(\cdot)$ is the rate function defined as 
	$R_d(x) := -\log(d - df_d^{-1}(x)),$
	and 
	$f_d(x) := h(x) + {x\log(d^2-1)} - \log d.$
	
	In terms of smooth min-entropy, we have for any $\e \in (0,1]$
	\begin{equation}
	\label{eqn:hmin-sampling}
	\entHmin^{\e}(A_S|ES)_{\rho} \geqslant k R_d(h_{\min}) - \log(n^2+1) - \log\frac{2}{\e^2},
	\end{equation}
	where $h_{\min} := \frac{\hmin(A^n|E)_{\rho}}{n}$.
\end{thm}
See Figure \ref{fig:sampling-plots} for a plot of $R_2(h_2)$. Note that $f_d$ is an increasing function on $[0, \frac{d^2-1}{d^2}]$ with $f_d(0) = -\log d$ and $f_d\left(\frac{d^2-1}{d^2}\right) = \log d$. We can thus define its inverse function $f_d^{-1} : [-\log d, \log d] \to [0, \frac{d^2-1}{d^2}]$. 

\begin{figure}[h!]
	\centering
	\begin{tikzpicture}[scale=4]
		\draw[very thin,color=gray,step=0.25] (-1,-1) grid (1,1);
		\draw[->] (-1,0) -- (1,0) node[right] {$h_2$};
		\foreach \x in {-1, -0.75, -0.5, -0.25, 0, 0.25, 0.5, 0.75, 1}
			\node[below,font=\footnotesize] at (\x cm,0) {\x};
		\draw[->] (0,-1) -- (0,1) node[above] {Rate function};
		\foreach \y in {-1, -0.75, -0.5, -0.25, 0, 0.25, 0.5, 0.75, 1}
			\node[left,font=\footnotesize] at (0,\y cm) {\y};
		%\draw[blue] plot[smooth] file {sampling.dat};
		\draw[blue,thick] plot[domain=0.0001:0.75,variable=\t] ({-\t*log2(\t)-(1-\t)*log2(1-\t)+\t*log2(3)-1},{-log2(2-2*\t)});   %R_2
		%\draw[red,densely dashed,thick] plot[smooth] file {cqplot-d2.dat};
		\draw[red,densely dashed,thick] plot[smooth,domain=0.0001:0.5, variable=\t] ({-\t*log2(\t) - (1-\t)*log2(1-\t)}, {-log2(1-\t)});     %C_2
		%\draw[black] plot[smooth] file {hypercontractive.dat};
		%\draw[violet,dashdotted,thick] plot[smooth] file {juerg.dat};
		\draw[violet,dashdotted,thick] plot[smooth,domain=0.00001:0.11002, variable=\t] ({-2*\t*log2(\t) - 2*(1-\t)*log2(1-\t)},{\t/6});    %Jürg
		%\draw[brown,loosely dashed,thick] plot[smooth] file {upperbound-quantum.dat};
		\draw[brown,loosely dashed,thick] plot[domain=0.0001:3/8,variable=\t] ({-\t*log2(\t)-(1-\t)*log2(1-\t)+\t*log2(3)-1},{-log2(2-4*\t)});   %Q upper bound
		%\draw[teal,densely dotted,thick] plot[smooth] file {upperbound-classical.dat};
		\draw[teal,densely dotted,thick] plot[smooth,domain=0.0001:0.25, variable=\t] ({-\t*log2(\t) - (1-\t)*log2(1-\t)}, {-log2(1-2*\t)});     %C upper bound
		\draw[olive,loosely dashdotted,thick] plot[domain=0:1] (\x, \x - 0.3);    % Vadhan
	\end{tikzpicture}
	\caption{Plot of our quantum-quantum rate function $R_2(h_2)$ from Theorem \ref{thm:h2-sampling} (\protect\tikz[baseline=-0.5ex]{\protect\draw[blue] (0,0) -- (0.5,0);}), our classical-quantum rate function $C_2(h_2)$ from Theorem \ref{thm:cq-h2-sampling} (\protect\tikz[baseline=-0.5ex]{\protect\draw[red,densely dashed,thick] (0,0) -- (0.5,0);}), Wullschleger's min-entropy sampling result \cite[Corollary 1]{Wul10} (\protect\tikz[baseline=-0.5ex]{\protect\draw[violet,dashdotted,thick] (0,0) -- (0.5,0);}), Vadhan's purely classical min-entropy sampling results~\cite[Lemma 6.2]{Vad03} (\protect\tikz[baseline=-0.5ex]{\protect\draw[olive,loosely dashdotted,thick] (0,0) -- (0.5,0);}), and the classical and quantum upper bounds we get from a state that is uniform on strings of a fixed type analyzed in Theorems \ref{thm:upper-bound-rate-qq} and \ref{thm:upper-bound-rate-cq} (\protect\tikz[baseline=-0.5ex]{\protect\draw[teal,densely dotted,thick] (0,0) -- (0.5,0);}, \protect\tikz[baseline=-0.5ex]{\protect\draw[brown,loosely dashed,thick] (0,0) -- (0.5,0);}).
	As Vadhan's result requires a choice of parameters we chose $\tau = 0.1$, which yields a lower bound on the \emph{smooth} min-entropy, with smoothing parameter of the order of $10^{-6}$ for a block size of $n=10000$.
	}
	\label{fig:sampling-plots}
\end{figure}

\begin{proof}
We start by observing that \eqref{eqn:h2-sampling} directly implies \eqref{eqn:hmin-sampling}. This follows from the fact that $h_2 \geqslant h_{\min}$ (Lemma \ref{lem:hmin-h2-fullyquantum}) and Lemma \ref{lem:smoothhmin-h2}. 

We now prove \eqref{eqn:h2-sampling} by applying Theorem~\ref{thm:general-h2} for an appropriately chosen
map $\mathcal{M}$.
%Define $\cM_{A^n \to A^kS} (X) = \frac{1}{\sqrt{\binom{n}{k}}} \sum_{S \subseteq [n], |S| = k} \tr_{S^c} [X] \otimes \proj{S}$, for $X \in \cL(A^n)$, where the second register contains a classical description of the set $S$, and $S^c$ denotes the complement of $S$ in $[n]$. The reason for this normalization will be clear in the following calculation. Our first task is to relate this map to the task of sampling entanglement. We have
%\begin{align*}
%2^{-\entHtwo(A^kS|E)_{\cM(\rho)|\rho}} &= \tr\left[\left(\rho^{-1/4}_E (\cM \otimes \id)(\rho_{A^nE}) \rho_E^{-1/4} \right)^2\right] \\
%&= \tr\left[\left(\rho^{-1/4}_E \left(\frac{1}{\sqrt{\binom{n}{k}}} \sum_{|S|=k} \rho_{A_SE} \otimes \proj{S} \right) \rho_E^{-1/4} \right)^2\right] \\
%&= \mbE_{S \subseteq [n], |S|=k} \tr\left[\left(\rho^{-1/4}_E \rho_{A_SE} \rho_E^{-1/4} \right)^2\right] = 2^{-\entHtwo(A_S|ES)_{\rho}},
%\end{align*}
%where for the last equality we used the expression for the entropy conditioned on the classical system $S$ (Lemma \ref{lem:condition-cl}). Note that in the second line above, we slightly abused notation and identified $A^k$ with the spaces $A_S$ for different values of $S$.

%\alternative{
Define $\cM_{A^n \to A^kS} (X) = \frac{1}{\binom{n}{k}} \sum_{S \subseteq [n], |S| = k} \tr_{S^c} [X] \otimes \proj{S}$, for $X \in \cL(A^n)$, where the second register contains a classical description of the set $S$, and $S^c$ denotes the complement of $S$ in $[n]$. Our first task is to relate this map to the task of sampling entanglement. We have
\begin{align*}
2^{-\entHtwo(A^kS|E)_{\cM(\rho)|\rho}} &= \tr\left[\left(\rho^{-1/4}_E (\cM \otimes \id)(\rho_{A^nE}) \rho_E^{-1/4} \right)^2\right] \\
&= \tr\left[\left(\rho^{-1/4}_E \left(\frac{1}{\binom{n}{k}} \sum_{|S|=k} \rho_{A_SE} \otimes \proj{S} \right) \rho_E^{-1/4} \right)^2\right] \\
&= 2^{-\entHtwo(A_SS|E)_{\rho}} \\
&= \frac{1}{\binom{n}{k}} \mbE_{S \subseteq [n], |S|=k} \tr\left[\left(\rho^{-1/4}_E \rho_{A_SE} \rho_E^{-1/4} \right)^2\right] = \frac{1}{\binom{n}{k}} 2^{-\entHtwo(A_S|ES)_{\rho}},
\end{align*}
where for the last equality we used the expression for the entropy conditioned on the classical system $S$ (Lemma \ref{lem:condition-cl}). This last equality can be seen as a chain rule $\entHtwo(AY|E) = \entHtwo(A|EY) - \log |Y|$ for states of the form $\rho_{AEY} = \frac{1}{|Y|}\sum_{y} \rho_{AE}(y) \otimes \proj{y}$ with $\tr_A[\rho_{AE}(y)] = \tr_A[\rho_{AE}]$ for all $y$. 
%Note that in the second line above, we slightly abused notation and identified $A^k$ with the spaces $A_S$ for different values of $S$. 
Observe also that we could have chosen the map $\cM_{A^n \to A^kS} (X) = \frac{1}{\sqrt{\binom{n}{k}}} \sum_{S \subseteq [n], |S| = k} \tr_{S^c} [X] \otimes \proj{S}$, which is not trace preserving, to get more directly $2^{-\entHtwo(A^kS|E)_{\cM(\rho)|\rho}} = 2^{-\entHtwo(A_S|ES)_{\rho}}$ and the result would be exactly the same. We preferred to use the physical (i.e., normalized) sampling map to make the steps of the proof more transparent.
%}

%Our second task is to show that our choice of map satisfies the conditions of Theorem~\ref{thm:general-h2}.
%%We now evaluate the state $\cM^{\dagger} \circ \cM (\Phi_{A^n\bar{A}^n})$. 
%We have 
%\begin{align*}
%	( (\cM^{\dagger} \circ \cM) \otimes \id_{\bar{A}^n}) (\Phi_{A^n \bar{A}^n}) &= \cM^{\dagger}\left( \frac{1}{\sqrt{\binom{n}{k}}} \sum_{|S|=k} \proj{S} \otimes \Phi_{A_S\bar{A}_S} \otimes \id_{\bar{A}_{S^c}} \right) \\
%&= \frac{1}{\binom{n}{k}} \sum_{|S|=k} \Phi_{A_S \bar{A}_S} \otimes \id_{A_{S^c}\bar{A}_{S^c}}.
%\end{align*}
%We now write this operator in terms of $\{\Phi_s\}_{s \in [d^2]^n}$. Recall that $\{\frac{1}{\sqrt{d^n}} \ket{\Phi_s}\}_s$ forms an orthonormal basis and thus $\id_{A^n \bar{A}^n} = \frac{1}{d^n}\sum_{s \in [d^2]^n} \Phi_s$:
%\begin{align*}
%	( (\cM^{\dagger} \circ \cM) \otimes \id_{\bar{A}^n}) (\Phi_{A^n \bar{A}^n}) &= \frac{1}{d^{n-k} \binom{n}{k}} \sum_{|S| = k} \sum_{s : \supp(s) \subseteq S^c} \Phi_s \\
%&= \frac{1}{d^{n-k}\binom{n}{k}} \sum_{s : |s| \leqslant n-k} \binom{n-|s|}{k} \Phi_s. 
%\end{align*}
%As a result, the coefficients $\lambda_s$ from Theorem \ref{thm:general-h2} are $\lambda_s = \frac{\binom{n-|s|}{k}}{d^{n-k}\binom{n}{k}}$. Observe that $\lambda_s$ only depends on $|s|$ and is a decreasing function of $|s|$. In order to apply Theorem \ref{thm:general-h2}, it is natural to choose the partition $\mfS_+ \cup \mfS_-$ of the form $\mfS_+ = \{ s \in [d^2]^n : |s| \leqslant \ell_0\}$ and $\mfS_- = \{s \in [d^2]^n : |s| > \ell_0\}$ for a value of $\ell_0 \in \{0, \dots, n\}$ to be chosen as a function of $h_2$.

Our second task is to show that our choice of map satisfies the conditions of Theorem~\ref{thm:general-h2}.
%We now evaluate the state $\cM^{\dagger} \circ \cM (\Phi_{A^n\bar{A}^n})$. 
We have 
\begin{align*}
	( (\cM^{\dagger} \circ \cM) \otimes \id_{\bar{A}^n}) (\Phi_{A^n \bar{A}^n}) &= \cM^{\dagger}\left( \frac{1}{\binom{n}{k}} \sum_{|S|=k} \proj{S} \otimes \Phi_{A_S\bar{A}_S} \otimes \id_{\bar{A}_{S^c}} \right) \\
&= \frac{1}{\binom{n}{k}^2} \sum_{|S|=k} \Phi_{A_S \bar{A}_S} \otimes \id_{A_{S^c}\bar{A}_{S^c}}.
\end{align*}
We now write this operator in terms of $\{\Phi_s\}_{s \in [d^2]^n}$. Recall that $\{\frac{1}{\sqrt{d^n}} \ket{\Phi_s}\}_s$ forms an orthonormal basis and thus $\id_{A^n \bar{A}^n} = \frac{1}{d^n}\sum_{s \in [d^2]^n} \Phi_s$:
\begin{align*}
	( (\cM^{\dagger} \circ \cM) \otimes \id_{\bar{A}^n}) (\Phi_{A^n \bar{A}^n}) &= \frac{1}{d^{n-k} \binom{n}{k}^2} \sum_{|S| = k} \sum_{s : \supp(s) \subseteq S^c} \Phi_s \\
&= \frac{1}{d^{n-k}\binom{n}{k}^2} \sum_{s : |s| \leqslant n-k} \binom{n-|s|}{k} \Phi_s. 
\end{align*}
As a result, the coefficients $\lambda_s$ from Theorem \ref{thm:general-h2} are $\lambda_s = \frac{\binom{n-|s|}{k}}{d^{n-k}\binom{n}{k}^2}$. Observe that $\lambda_s$ only depends on $|s|$ and is a decreasing function of $|s|$. In order to apply Theorem \ref{thm:general-h2}, it is natural to choose the partition $\mfS_+ \cup \mfS_-$ of the form $\mfS_+ = \{ s \in [d^2]^n : |s| \leqslant \ell_0\}$ and $\mfS_- = \{s \in [d^2]^n : |s| > \ell_0\}$ for a value of $\ell_0 \in \{0, \dots, n\}$ to be chosen as a function of $h_2$.

Writing equation \eqref{eqn:general-h2} in our case we obtain,
\begin{align}
2^{-\htwo(A_S|ES)_{\rho}} &\leqslant \sum_{\ell=0}^{\ell_0} \frac{\binom{n-\ell}{k}}{d^{n-k}\binom{n}{k}} \binom{n}{\ell}(d^2-1)^{\ell} 2^{-h_2 n} +  \frac{\binom{n-\ell_0-1}{k}}{\binom{n}{k}} d^k \notag \\
&= \frac{2^{-h_2 n}}{d^{n-k}}\sum_{\ell=0}^{\ell_0} \binom{n-k}{\ell} (d^2-1)^{\ell} +  \frac{\binom{n-\ell_0-1}{k}}{\binom{n}{k}} d^k. \label{eqn:sum-split}
\end{align}
Now all that remains is to optimize over $\ell_0$ and to find a simple expression for this quantity. Before choosing $\ell_0$, we simplify the expression above. For the second term, we bound
\[
\frac{\binom{n-\ell_0-1}{k}}{\binom{n}{k}} d^k \leqslant \left(\frac{n-\ell_0-1}{n}\right)^k d^k.
\]
To obtain a simple bound on the first term, we use the following lemma, which is proven in the appendix.
\begin{lem}[label=lem:binomial-sum, restate=RestateBinomialSum]
For any $\ell_0 \in \{0, \dots, n\}$ such that $\ell_0 \leqslant \frac{d^2-1}{d^2} n$ where $d^2 < n$, we have
\[
\sum_{\ell=0}^{\ell_0} {n-k \choose \ell} (d^2-1)^{\ell}  \leqslant n^2 {\binom{n}{\ell_0} (d^2-1)^{\ell_0}} \mathrm{max} \left(\frac{n-\ell_0-1}{n}, \frac{1}{d^2} \right)^k.
\]
\end{lem}

It then follows from equation \eqref{eqn:sum-split} that
\begin{align*}
2^{-\entHtwo(A_S|ES)_{\rho}}
&\leqslant \max\left(\frac{n-\ell_0-1}{n}, \frac{1}{d^2}\right)^k d^k \left( \frac{2^{-h_2n}}{d^n} n^2 \binom{n}{\ell_0} (d^2-1)^{\ell_0} + 1\right).
\end{align*}

%
%
%% Choice of \ell_0
%
%

We now determine the value of $\ell_0$ as a function of $h_2$.
Observe that using Lemma \ref{lem:sum-binomial}, we have $\binom{n}{\ell} (d^2-1)^{\ell} \leqslant 2^{nh(\ell_0/n)}(d^2-1)^{\ell_0} = 2^{n f_d(\ell_0/n)} d^n$ provided $\ell_0 \leqslant \frac{d^2-1}{d^2}n$. 
We define $\ell_0$ to be the largest integer that is at most $\frac{d^2-1}{d^2}n$ such that $f_d(\ell_0/n) \leqslant h_2$. As a result, we have
\begin{align}
\label{eqn:bound-in-terms-l0}
2^{-\entHtwo(A_S|ES)_{\rho}}
&\leqslant \max\left(\frac{n-\ell_0-1}{n} , \frac{1}{d^2} \right)^k d^k \left(n^2 + 1\right).
\end{align}
Observe also that in the case where the maximum is $1/d^2$, the result follows directly as $R_d(h_2) \leq \log d$. In the case where $(n-\ell_0-1)/n > 1/d^2$, we observe that $(\ell_0+1)/n > f_d^{-1}(h_2)$ by our choice of $\ell_0$. Note that if $\ell_0+1 \leqslant (d^2-1)/d^2 \cdot n$, this follows from the fact that $f_d$ is nondecreasing, and otherwise it follows from the fact that by definition $f_d^{-1}$ is always upper bounded by $(d^2-1)/d^2$. 
We now write $\left(\frac{n-\ell_0-1}{n}\right)^k$ in terms of the entropy rate $h_2$:
\begin{align*}
	k \log \left( \frac{n-\ell_0-1}{n} \right) &= k\log\left( 1- \frac{\ell_0+1}{n} \right)\\
&\leqslant k\log(1 - f_d^{-1}(h_2))\\
&= k \log(d - df_d^{-1}(h_2)) - k\log d\\
&= -k R_d(h_2) - k\log d.
\end{align*}
By plugging these inequalities into \eqref{eqn:bound-in-terms-l0}, we obtain the desired result. 
\end{proof}

\subsubsection{An upper bound on the rate function}\label{sec:qq-upper-bound}
Note that the rate function obtained in Theorem \ref{thm:h2-sampling} is independent of the state $\rho_{AE}$ and of the size of the sample $k$. The objective of this section is to show that with such a requirement, the rate function $R_d$ cannot be improved too much especially when $h_2$ is close to the minimal value of $-\log d$.

\begin{defin}
	We define the optimal rate function $R_d^{\opt}$ as
	\[ R_d^{\opt}(h_2) :=  \liminf_{n \geqslant 1} \left( \min_{k \in [n], \rho_{A^nE} \text{ such that } \frac{1}{n} \htwo(A^n|E) \geqslant h_2} \frac{1}{k} \htwo(A_S|ES)_{\rho} \right)\ , \]
	where $A^n=A_1,\ldots,A_n$ is comprised of $n$ qudits of dimension $d$.
\end{defin} 
We now derive an upper bound on the rate function that will show that our result is fairly close to optimal for small $h_2$ and small $k$. The idea is to choose a state that consists of $n$ EPR pairs that have been corrupted by a fixed-weight generalized Pauli error. In this case, if this weight is small enough, the sample will avoid all the errors with relatively large probability and the collision entropy of the sampling is going to be much smaller than $kh_2$.
%a low-weight error contributes a lot to $h_2$, but the errors are rare enough that they have a low probability of being included in the sample, and we thus get that the average entropy in the sample is much lower than $kh_2$.

\begin{thm}
\label{thm:upper-bound-rate-qq}
	It holds that $R_d^{\opt}(h_2) \leqslant -\log(d - 2df_d^{-1}(h_2))$.
\end{thm}
\begin{proof}
	Let $E = B^{n} \cong A^n$, and consider the state \[ \rho_{A^nB^n} = \left( {n \choose w} (d^2-1)^w \right)^{-1} \sum_{s, |s|=w} \frac{\Phi_s}{d^n} \]
	for some particular $w$. This is a maximally entangled state between $A^n$ and $B^n$ that has been corrupted by a random error of weight exactly $w$. 
	We can compute its collision entropy
	\begin{align*}
		2^{-\htwo(A^n|B^n)_{\rho}} = \tr[\rho_{B^n}^{-1/2}\rho_{A^nB^n} \rho_{B^n}^{-1/2} \rho_{A^nB^n}] &= \left( {n \choose w} (d^2-1)^w \right)^{-2} \sum_{s, |s|=w} \tr[\ident_{A^n}]\\
		&= \left( {n \choose w} (d^2-1)^w \right)^{-1} d^n\\
		&\leqslant 2^{-nh(w/n) - w\log(d^2-1) + n\log d + \log n}.
	\end{align*}
	Hence, $h_2 \geqslant f_d(w/n) - \frac{1}{n}\log n$.

	Now, let us compute the collision entropy for a random subsystem of size $k$ with $\frac{k}{n} \rightarrow 0$. Note that we have 
	\begin{align*}
		\rho_{A_SB^n} &= \sum_{s \in [d^2]^{S}}  \pr{\sigma_S = s} \frac{\Phi_s \otimes \ident_{B_{S^c}}}{d^n},
	\end{align*}
	where $\sigma \in [d^2]^n$ is a random string of weight exactly $w$, and $\sigma_S$ is the substring index by elements of $S$. Then we have 
	\begin{align*}
	\tr\left[\left(\rho_{B^n}^{-1/4} \rho_{A_SB^n} \rho_{B^n}^{-1/4}\right)^2\right] &= \sum_{s,s' \in [d^2]^{S}}  \pr{\sigma_S = s} \pr{\sigma'_S = s'} \frac{\tr[\Phi_s \Phi_{s'}]}{d^{k}} \\
			&= d^k \pr{\sigma_S = \sigma'_S}.
	\end{align*}
	Thus, we want to evaluate the average over $S$ of this probability. For any fixed $\sigma$ and $\sigma'$ of weight $w$ and choosing a random subset of size $k$, the corresponding substrings will be the same if they avoid all positions where $\sigma$ or $\sigma'$ are non-zero. As a result the collision probability for the sample is at least
	% It easy to show that only collisions of type at most $2w$ can occur, which means that the collision probability of the $k$-qudit sample will be at least 
	\begin{align*}
		\frac{n-2w}{n} \cdots \frac{n-2w-k+1}{n-k+1}
		&\geqslant \left( \frac{n-k-2w}{n-k} \right)^k \\
		&= \left( 1 - 2\left( \frac{w}{n} \right) \left( \frac{n}{n-k} \right) \right)^k\\
		&\geqslant \left( 1 - 2f_d^{-1}\left(h_2 + \frac{1}{n}\log n\right) \left( \frac{n}{n-k} \right) \right)^k.
	\end{align*}
	Taking the limit over $n \rightarrow \infty$ and $\frac{k}{n} \rightarrow 0$, we get that
	\[ 2^{-\htwo(A_S|E)} \geqslant (d(1 - 2f_{d}^{-1}(h_2)))^k. \]	
	This directly yields the theorem.
\end{proof}

\subsubsection{Applications of entanglement sampling}\label{sec:decouple}

An immediate consequence of our result on entanglement sampling concerns the existence of decouplers (QQ-extractors) using only very few qubits. A decoupling operation
is some process $\mathcal{K}_{A \rightarrow B}$ that applied to the $A$ system transforms $\rho_{AE}$ to a state that is close to $\tau_B \otimes \rho_E$, where
$\tau_{B}$ is a state that depends only on the map $\mathcal{K}$ but not on the initial state $\rho_{AE}$. In quantum information theory, such processes typically consist of applying a random unitary $U$ to $A$, followed by a map $\mathcal{T}_{A \rightarrow B}$ such as the partial trace operation. That is, the map $\mathcal{K}$ is of the form
$\mathcal{K}(\rho_A) = \int d\mu(U) \mathcal{T}_{A \rightarrow B}(U \rho_A U^\dagger) \otimes \proj{U}$, where $\ket{U}$ is a classical register containing the choice of unitary.

Decoupling theorems in quantum information theory have their origin in quantum channel coding~\cite{HOW06,ADHW09,HHYW08} where $\mathcal{T}$ is usually the partial trace, and $\tau_B = \id/|B|$. 
In this context, the size of the system $|B|$ that one can decouple from $E$, can be related to the number of qubits that one can pass through a quantum channel whose environment 
is $E$ with vanishing error. In this context, the choice of unitary $U$ yields an encoding scheme (see~\cite{Dup09} for details). 
More recently, the decoupling theorem has been generalized to a wide variety of maps $\mathcal{T}$~\cite{Dup09,DBWR10}.

Decoupling results are known when the unitaries are chosen from the Haar measure~\cite{ADHW09,HHYW08,Dup09,DBWR10}, from a $2$-design, from an approximate $2$-design~\cite{sdtr11}, or from even more efficient sets of unitaries~\cite{BF13}. In contrast, when $A$ is classical, many decoupling operations are known in the form of randomness extractors discussed in the introduction (see~\cite{vadhan:survey} for a survey). Of particular interest in both computational~\cite{HHYW08} and physical applications~\cite{HP07,lidia:inprep,renato:workGain, adrian:thesis} are unitaries which are efficient. In a computational setting, this generally refers to unitaries that can be implemented using low-depth quantum circuits, whereas in physical scenarios it is usually of interest that they arise from Hamiltonians involving only nearest neighbour interactions over a short period of time. 

As an example of the physical relevance of decoupling theorems, let us consider the case where $A$ is comprised of a system $A_{\rm sys}$ and a bath $A_{\rm bath}$, 
and $\mathcal{T} = \tr_{A_{\rm bath}}$ is the operation that traces out the bath. A decoupling theorem for certain classes of unitaries then says that for very many unitaries
in that set, the resulting state of the system $\tau_B$ is independent of its initial state, one of the steps considered in the process of thermalization~\cite{linden:thermalize}.
That is, it tells us that certain evolutions of the system and the bath, namely those corresponding to such unitaries, can lead to thermalization. This holds even in the stronger sense
of relative thermalization where one requires that the state of the system becomes independent of an observer holding $E$~\cite{lidia:inprep}.
In fact, the decoupling theorem~\cite{DBWR10} for Haar measure random unitaries can even be used~\cite{adrian:thesis} to recover the results of~\cite{short:pick} stating that for most initial states of $A$, or equivalently most unitary evolutions on $A$, the resulting state is close to the canonical state. 

As such, it is an interesting question to determine which sets of unitaries lead to a decoupling theorem. Here, our goal is to show that if $A^n = A_1,\ldots,A_n$ consists of $n$ qudits, then there exist decoupling operations involving only a (small) subset of such qudits. As outlined in the introduction, one generic way to accomplish this task is to show that the fully quantum min-entropy can be sampled. Decoupling operations involving only few qudits can then be obtained in a ''sample-then-decouple'' fashion similar to the classical ''sample-then-extract'' approach of~\cite{Vad03}. That is, one first samples a set of qubits, and then applies an arbitrary decoupling operation on the resulting sample.

Our result extends to any of the more modern decoupling theorems involving entropy measures~\cite{Dup09,DBWR10}.~\footnote{In contrast to statements involving 
only the dimensions of systems as in e.g.,~\cite{ADHW09}.} To illustrate this idea, let us consider the example of $A^n = A_1,\ldots,A_n$ consisting of $n$ qubits, unitaries chosen from the Haar measure, and
$\mathcal{T}$ being the partial trace operation $\tr_{n-r}(\rho_A)$ tracing out all but $r$ of the $n$ qubits. In terms of the $\htwo$ entropy is was shown~\cite{Dup09,DBWR10}
that
\begin{align}
	\int d(U) \left\|\tr_{n-r} \otimes \id_E(\rho_{AE}) - \frac{\id}{2^r} \otimes \rho_E \right\|_1 \leqslant 2^{-\demi (\htwo(A|E) + n - 2r)}\ ,
\end{align}
where $\|\rho - \sigma\|_1$ is the trace distance of $\rho$ and $\sigma$. 
If we now first sample a subset of size $k$ of the qubits, then our sampling result states that for unitaries chosen according to the Haar measure of qubits
\begin{align}
	\int d(U) \left\|\tr_{|S|-r} \otimes \id_E(\rho_{AES}) - \frac{\id}{2^r} \otimes \rho_{ES} \right\|_1 \leqslant 2^{-\frac{1}{2} 
	\left[|S|\left(R_2\left(\frac{\htwo(A|E)}{n}\right) - 1\right) - 2r - \log (n^2+1)\right]}\ ,
\end{align}
for the rate function given in Theorem~\ref{thm:h2-sampling}. Similarly, our sampling result can be applied to the special kinds of decoupling maps known as quantum-to-classical randomness extractors~\cite{BFW12}. In this context, sampling allows the generation of classical randomness from a quantum system~\footnote{Of which we only have a guarantee about the entropy.} 
by applying measurements to only a few of the qubits of $A$.

%\steph{I'm making an example here even though its obvious to us, because I do not think that any of the more physically motivated crowd who may now be interested in decoupling theorems understands this otherwise - it's already a stretch if they made it to this part of the manuscript :)}

\subsection{Classical-quantum min-entropy sampling}\label{sec:classical-quantum}

\subsubsection{Statement}
Observe that in the case where the system $A^n$ is classical, i.e., $\rho_{A^n E} = \sum_{x^n \in [d]^n} p(x^n) \proj{x^n} \otimes \rho_E(x^n)$ for some distribution $p$ and states $\rho_E(x^n)$, Theorem \ref{thm:h2-sampling} can still be applied but in many cases it give trivial bounds. In fact, when $A^n$ is classical, we have $\htwo(A^n|E) \geqslant 0$ as well as $\htwo(A_S|ES) \geqslant 0$. In order to improve on the lower bound of Theorem \ref{thm:h2-sampling} in the case of a classical system, we can apply Theorem \ref{thm:general-h2} to a more specific map $\cM$ that \emph{measures} the systems $A_S$ that are sampled. This allows us to obtain a lower bound on the collision entropy $\htwo(A_S|ES)$ that is nontrivial for the entire range $\htwo(A^n|E) \in [0, n \log d]$.

Unlike the fully quantum case about which not much was known, the classical-quantum min-entropy sampling has been previously studied in particular in \cite{KR07,Wul10,BARdW08}. We briefly highlight the similarities and differences with our results in Theorem \ref{thm:cq-h2-sampling}. The work of \cite{BARdW08} is restricted to the case where $A^n$ is uniformly distributed and obtains a lower bound on the non-smoothed min-entropy~\footnote{The fact that the min-entropy is non-smoothed is important for the application to random access codes.} of the sample as a function of the dimension of the system $E$ rather than the conditional entropy. This special case is particularly interesting in the context of random access codes. The parameters they obtain are better when the dimension of $E$ is small, i.e., $h_2$ is large. However, their techniques fail to give a nontrivial bound when $h_2$ is small. See Section \ref{sec:qrac} for more details. The sampling theorem of \cite{Wul10} works for general classical-quantum states and gives a lower bound on the non-smoothed min-entropy of the sample. The parameters are illustrated in Figure \ref{fig:sampling-plots}. The work of \cite{KR07} considers the general classical-quantum case and focuses on the smoothed min-entropy. This result extends Vadhan's classical min-entropy sampling~\cite{Vad03} result to the case of quantum side information. Hiding technicalities (like the fact one should sample blocks rather than bits) and neglecting terms that depend on the smoothing parameters, the rate function they obtain is basically optimal $R(\alpha) = \alpha$, \footnote{To obtain such a result, smoothing is in fact necessary as shown by the example of Theorem \ref{thm:upper-bound-rate-cq}.} as the plot of Vadhan's result in Figure \ref{fig:sampling-plots}.

Our sampling result has an application to randomness extraction, in that it yields a general way to construct locally computable extractors even with respect to quantum side information $E$. This is analogous to the application of entanglement sampling to decoupling discussed above.

%\fred{Do we also need a condition regarding $n$ being bigger than $d^2$ (or perhaps $d$?).}
\begin{thm}%[name=Classical-quantum min-entropy sampling, label=thm:cq-h2-sampling, restate=RestateCQsampling]
\label{thm:cq-h2-sampling}
	Let $\rho_{A^nE}$ be a classical-quantum state, and $1 \leqslant k\leqslant n$, let $d=|A|$, and let $h_2 := \frac{\htwo(A^n|E)_{\rho}}{n}$. Then, for any $n > d$,
	\[ 2^{-\htwo(A_S|ES)_{\rho}} = \mbE_{S \subseteq [n], |S|=k} 2^{-\htwo(A_S|E)_{\rho}} \leqslant 2^{-kC_d(h_2) + \log (n^2+1)}, \]
	where $C_d(\cdot)$ is the rate function defined as 
	$C_d(\alpha) := -\log(1 - c_d^{-1}(\alpha)), $
	and
	$c_d(\alpha) := h(\alpha) + {\alpha \log(d-1)}.$
		In terms of smooth min-entropy, we have for any $\e \in (0,1]$
	\begin{equation}
	\label{eqn:hmin-sampling-cq}
	\entHmin^{\e}(A_S|ES)_{\rho} \geqslant k C_d(h_{\min}) - \log(n^2+1) - \log\frac{2}{\e^2},
	\end{equation}
	where $h_{\min} := \frac{\hmin(A^n|E)_{\rho}}{n}$.
\end{thm}
See Figure \ref{fig:sampling-plots} for a plot of $C_2(h_2)$. Note that $c_d$ is an increasing function on $[0, \frac{d-1}{d}]$ with $c_d(0) = 0$ and $c_d\left(\frac{d-1}{d}\right) = \log d$. The inverse function $c_d^{-1} : [0, \log d] \to [0, \frac{d-1}{d}]$ is therefore well-defined. 
\begin{proof}
The proof is very similar to that of Theorem \ref{thm:h2-sampling}: one uses Theorem \ref{thm:general-h2} with $\cM_{A^n \to A^kS} (X) = \frac{1}{\binom{n}{k}} \sum_{S \subseteq [n], |S| = k} \sum_{x^k \in [d]^S} \bra{x^k}\tr_{S^c} [X] \ket{x^k} \otimes \proj{x^k} \otimes \proj{S}$, where $\{\ket{x^k}\}_{x^k \in [d]^S}$ is the standard basis of $A_S$. We then also have in this case $2^{-\entHtwo(A^k S | E)_{\cM(\rho)|\rho}} = \frac{1}{\binom{n}{k}} 2^{-\entHtwo(A_S|ES)_{\rho}}$.
Moreover, to apply the theorem, we compute
\[
((\cM^{\dagger} \circ \cM) \otimes \id_{\bar{A}^n} )(\Phi_{A^n \bar{A}^n}) = \frac{1}{\binom{n}{k}^2} \sum_{|S|=k} \left( \sum_{x^k} \proj{x^k}_{A_S} \otimes \proj{x^k}_{\bar{A}_S} \right) \otimes \id_{A_{S^c}\bar{A}_{S^c}}.
\]
Recall that we want to write this operator in terms of $\Phi_s = (W_s \otimes \id) \Phi_{A^n \bar{A}^n} (W_s^{\dagger} \otimes \id)$, where $W_s = W_{s_1} \otimes \cdots \otimes W_{s^n}$ is a product of generalized Pauli operators. Let us now assume that the numbering of the Pauli operators is such that $W_0, \dots, W_{d-1}$ are defined as $W_{y} \ket{x} = e^{2\pi i x y / d} \ket{x}$ for all $x,y \in [d]$. It then follows that
\begin{align*}
\frac{1}{d} \sum_{y \in [d]} \Phi_y &= \frac{1}{d} \sum_{y \in [d]} \sum_{x,x' \in [d]} e^{2\pi i (x-x') y / d} \ket{x} \bra{x'} \otimes \ket{x} \bra{x'} \\
&= \sum_{x} \proj{x}_{A} \otimes \proj{x}_{\bar{A}}.
\end{align*}
As a result, we can write
\begin{align*}
((\cM^{\dagger} \circ \cM) \otimes \id_{\bar{A}^n} )(\Phi_{A^n \bar{A}^n}) &= \frac{1}{\binom{n}{k}^2 d^n} \sum_{|S|=k} \sum_{\substack{s \in [d^2]^n \\  s_i \in [d], i \in S }} \Phi_s \\
&= \frac{1}{\binom{n}{k}^2 d^n} \sum_{s : |s|_{< d} \leqslant n-k} \binom{n-|s|_{<d}}{k} \Phi_s\ ,
\end{align*}
where $|s|_{<d} = |\{i \in [n] : s_i \in [d]\}|$.
As a result, the coefficients $\lambda_s$ from Theorem \ref{thm:general-h2} are $\lambda_s = \frac{\binom{n-|s|_{<d}}{k}}{d^{n}\binom{n}{k}^2}$, which only depends on $|s|_{<d}$ and is a decreasing function of $|s|_{<d}$. As before, it is natural to choose the partition $\mfS_+ \cup \mfS_-$ from Theorem \ref{thm:general-h2} of the form $\mfS_+ = \{ s \in [d^2]^n : |s|_{<d} \leqslant \ell_0\}$ and $\mfS_- = \{s \in [d^2]^n : |s|_{<d} > \ell_0\}$ for a value of $\ell_0 \in \{0, \dots, n\}$ to be chosen as a function of $h_2$. Applying Theorem \ref{thm:general-h2}, we have
\begin{align}
2^{-\htwo(A_S|ES)_{\rho}} &\leqslant \sum_{\ell=0}^{\ell_0} \frac{\binom{n-\ell}{k}}{d^{n}\binom{n}{k}} \binom{n}{\ell}(d^2-d)^{\ell} d^{n-\ell} 2^{-h_2 n} +  \frac{\binom{n-\ell_0-1}{k}}{\binom{n}{k}} \notag \\
&\leqslant 2^{-h_2 n}\sum_{\ell=0}^{\ell_0} \binom{n-k}{\ell} (d-1)^{\ell} +  \left(\frac{n-\ell_0 - 1}{n}\right)^k. \label{eqn:sum-split-cq}
\end{align}
%Now all that remains is to optimize over $\ell_0$ and to find a simple expression for this quantity. 
To obtain a simple bound on the first term, we use the same Lemma \ref{lem:binomial-sum} as in the proof of Theorem~\ref{thm:h2-sampling} replacing $d^2$ by $d$. Equation \eqref{eqn:sum-split-cq} then becomes
\begin{align*}
2^{-\entHtwo(A_S|ES)_{\rho}}
&\leqslant \max\left(\frac{n-\ell_0-1}{n}, \frac{1}{d}\right)^k \left( 2^{-h_2n} n^2 \binom{n}{\ell_0} (d-1)^{\ell_0} + 1\right).
\end{align*}

%
%
%% Choice of \ell_0
%
%

We now determine the value of $\ell_0$ as a function of $h_2$.
Observe that using Lemma \ref{lem:sum-binomial}, we have $\binom{n}{\ell_0} (d^2-1)^{\ell_0} \leqslant 2^{nh(\ell_0/n)}(d-1)^{\ell_0} = 2^{n c_d(\ell_0/n)}$ provided $\ell_0 \leqslant \frac{d-1}{d}n$. 
We define $\ell_0$ to be the largest integer that is at most $\frac{d-1}{d}n$ such that $c_d(\ell_0/n) \leqslant h_2$. As a result, we have
\begin{align}
\label{eqn:bound-in-terms-l0-cq}
2^{-\entHtwo(A_S|ES)_{\rho}}
&\leqslant \max\left(\frac{n-\ell_0-1}{n}, \frac{1}{d} \right)^k \left(n^2 + 1\right).
\end{align}
If the maximum is $1/d$, then we directly get the desired result.
Now we use the maximality of $\ell_0$ to say that $(\ell_0+1)/n > c_d^{-1}(h_2)$. 
%Note that in the case where $\ell_0+1 \leqslant (d^2-1)/d^2 \cdot n$, this follows from the fact that $f_d$ is nondecreasing, and otherwise it follows from the fact that by definition $f_d^{-1}$ is always upper bounded by $(d^2-1)/d^2$. 
%We now write $\left(\frac{n-\ell_0-1}{n}\right)^k$ in terms of the entropy rate $h_2$:
Finally,
\begin{align*}
	k \log \left( \frac{n-\ell_0-1}{n} \right) &= k\log\left( 1- \frac{\ell_0+1}{n} \right)\\
&\leqslant k\log(1 - c_d^{-1}(h_2))\\
&= -k C_d(h_2).
\end{align*}
By plugging these inequalities into \eqref{eqn:bound-in-terms-l0-cq}, we obtain the desired result. 
\end{proof}

\subsubsection{An upper bound on the classical rate function}\label{sec:cq-upper-bound}
Like in the quantum case, one can find an upper bound for the rate function. Here, our upper bound will even hold for non-conditional entropy (i.e., when $E$ is trivial).
\begin{defin}
	We define the optimal classical rate function $C_d^{\opt}$ as
	\[ C_d^{\opt}(h_2) :=  \liminf_{n \geqslant 1} \left( \min_{k \in [n], \rho_{X^nE} \text{ such that } \frac{1}{n} \htwo(X^n|E) \geqslant h_2} \frac{1}{k} \htwo(X_S|ES)_{\rho} \right), \]
	where $X^n=X_1,\ldots,X_n$ is comprised of $n$ dits of dimension $d$.
\end{defin}
We will now derive an upper bound on the rate function that will show that our result is fairly close to optimal for small $h_2$ and small $k$. We will derive our upper bound by considering the uniform distribution over strings of fixed Hamming weight. As in the fully quantum case, it will turn out that distributions of small Hamming weight still have a relatively high $h_2$ compared to the probability of getting a 0 in the sample, and this yields an average entropy for the sample that is much lower than $kh_2$.

\begin{thm}
\label{thm:upper-bound-rate-cq}
	It holds that $C_d^{\opt}(h_2) \leqslant -\log(1 - 2c_d^{-1}(h_2))$.
\end{thm}
\begin{proof}
	Let $E$ be trivial, and consider the state \[ \rho_{X^n} = |\{s, |s|=w \}|^{-1} \sum_{s, |s|=w} \ket{s}\bra{s} \]
	for some particular $w$. We can compute its collision entropy:
	\begin{align*}
		\tr[\rho_{X^n}^2] &= |\{s, |s|=w\}|^{-1}\\
		&= {n \choose w}^{-1} (d-1)^{-w}\\
		&\leqslant 2^{-nh(w/n) - w\log(d-1) + \log n}.
	\end{align*}
	Hence, $h_2 \geqslant c_d(w/n) - \frac{1}{n}\log n$.

	Now, let us compute the collision entropy of the sample when $\frac{k}{n} \rightarrow 0$. Fix a pair of strings $s$ and $s'$ with weight $w$. Choosing a random subset of size $k$, the corresponding substrings will be the same if they avoid all positions where $s$ or $s'$ are non-zero. As a result the collision probability for the sampled substring is at least
	% It easy to show that only collisions of type at most $2w$ can occur, which means that the collision probability of the $k$-qudit sample will be at least 
	\begin{align*}
		\frac{n-2w}{n} \cdots \frac{n-2w-k+1}{n-k+1}
		&\geqslant \left( \frac{n-k-2w}{n-k} \right)^k \\
		&= \left( 1 - 2\left( \frac{w}{n} \right) \left( \frac{n}{n-k} \right) \right)^k\\
		&\geqslant \left( 1 - 2c_d^{-1}\left(h_2 + \frac{1}{n}\log n\right) \left( \frac{n}{n-k} \right) \right)^k.
	\end{align*}
	Taking the limit over $n \rightarrow \infty$ and $\frac{k}{n} \rightarrow 0$, we get that
	\[ 2^{-\htwo(A_S|E)} \geqslant (1 - 2c_{d}^{-1}(h_2))^k. \]	
	This directly yields the theorem.
	
	\comment{
	\textbf{Remark in trying to obtain a better bound.} Note that in expectation the number of positions where both $s$ and $s'$ are non-zero is $w^2/n$. So for example assuming that the expectation and the median of this distribution coincide, we have that with probability at least $1/2$ on the choice of $s$ and $s'$, the number of positions where one of the strings is non-zero is $2w - w^2/n$. Then the bound we obtain is $C_d^{\opt}(h_2) \leqslant -\log(1 - 2c_d^{-1}(h_2) + \left(c_d^{-1}(h_2) \right)^2)$.
	}
%	Let us instead define the distribution as for each index $i \in [n]$, we choose independently a random non-zero element with probability $w/n$ and with probability $1-w/n$, we choose $0$. Then we have
%	\begin{align*}
%	\tr[\rho_{X^n}^2] &= \left( \left(1-\frac{w}{n}\right)^2 + \frac{w^2}{(d-1)n^2} \right)^n\\
%		&= {n \choose w}^{-1} (d-1)^{-w}\\
%		&\leqslant 2^{-nh(w/n) - w\log(d-1) + \log n}.
%	\end{align*}	
\end{proof}

\subsection{Dimension bounds for random access codes}
\label{sec:qrac}
One application of our sampling results is to bound the dimension of quantum random access codes.
To translate a result about min-entropy sampling into a result concerning random access codes, one simply considers the system $E$ to be composed of $m$ bits or qubits and 
then considers the special case of a uniform distribution on $X_1 \dots X_n$. That is, the state $\rho_{X^nE}$ is of the form
\begin{align}
	\rho_{X^nE} = \frac{1}{2^n} \sum_{x^n \in \01^n} \proj{x^n} \otimes \rho_E^{x^n}\ .
\end{align}
The quantity of interest when studying a random access encoding of a classical string $X^n$ is the minimal dimension of $E$ needed to recover any subset of size $k$ of the bits
with some desired probability $p$. Recall the operational interpretation of the conditional min-entropy $\hmin(X_S|ES)$ as the best probability for guessing the bitstring $X_S$ when having access to the system $E$. Thus, a lower bound on the min-entropy $\hmin(X_S|ES)$ directly gives an upper bound on the probability of successfully recovering a randomly chosen system $S$ of size $k$. The latter is exactly the success probability of $k$-out-of-$n$ random access code as defined in \cite{BARdW08}.

More precisely, using Lemma \ref{lem:hmin-h2}, Theorem \ref{thm:cq-h2-sampling} directly leads to a lower bound on the success probability $p$ of $k$-out-of-$n$ using $m$ qubits $p \leqslant \sqrt{2^{-kC_d(1-m/n) + \log (n^2+1)}}$. Compared to \cite{BARdW08}, this bound is better when $m$ is close to $n$. Specifically, when $m/n > 0.721$, the technique of \cite{BARdW08} does not give any probability bound. On the other hand, when $m/n$ becomes smaller, their probability bound becomes smaller. $k$-out-of-$n$ random access codes have also been studied in \cite{Wul10} and nontrivial upper bounds on the success probabilities are obtained for all values of $m$. The exponent of the success probability is illustrated in Figure \ref{fig:sampling-plots} (note that the plot for $C_d$ should be divided by two to interpret it as a guessing probability).

One could similarly define fully quantum random access codes. In this setting, we want to store $n$ qudits $A_1,\ldots,A_n$ of information into $m$ qudits so that a subset of $k$ qudits chosen at random can be recovered. Given $n$ and $m$, one can define the maximum average fidelity $F_{n,m,k}$ with which $k$ qudits can be recovered. As before, our goal will be to bound the dimension necessary to achieve a desired fidelity, or equivalently, establish an upper bound on the achievable fidelity as a function of the dimension.

\begin{thm}
Let $n > d^2$. For any $m \leqslant n$ and $1 \leqslant k \leqslant n$
\[
F^2_{n,m,k} \leqslant 2^{-\demi k \left( R_d\left(-\frac{m}{n} \log d \right) + \log d \right) + \demi \log (n^2+1)  } \\
\]
\end{thm} 
\begin{proof}
	Let $A^n$ be the system containing the $n$ qudits to be stored and $E$ be the $m$ qudits of storage. Then, for any $\rho_{A^nE}$, we have $\entHtwo(A^n|E) \geqslant -m\log d$. Using Theorem \ref{thm:h2-sampling} and Lemma \ref{lem:hmin-h2-fullyquantum}, we have
\begin{align*}
	2^{-\entHmin(A_S|ES)_{\sigma}} &= \mbE_{S \subseteq [n], |S|=k} 2^{-\hmin(A_S|E)_{\rho}} \\
&\leqslant \mbE_{S \subseteq [n], |S|=k} 2^{-\frac{1}{2}(\htwo(A_S|E)_{\rho} - k \log d)} \\
&\leqslant 2^{-\demi \left( kR_d\left(-\frac{m}{n} \log d \right) - \log (n^2+1) -k \log d \right)} \\
&\leqslant 2^{-\demi k \left( R_d\left(-\frac{m}{n} \log d \right) - \log d \right) + \demi \log (n^2+1),  } \\
\end{align*}
where $\sigma_{A^nES} = \rho_{A^nE} \otimes \frac{\ident_S}{\binom{n}{k}}$, with $S$ representing the choice of subset of $k$ qudits we want to recover. Now observe that $2^{-\entHmin(A_S|ES)} = 2^{k\log d} \max_{\cE_{ES} \to A'_S} F(\Phi_{A_SA'_S}^N, \id_{A_S} \otimes \cE(\rho_{ES}))^2$. The fidelity term is exactly the maximum fidelity with which the state on $A_S$ can be recovered from the system $E$.
\end{proof}

\subsection{High-order uncertainty relations against quantum side-information}

Uncertainty relations play a fundamental role in quantum information and in particular in quantum cryptography. Many of the modern security proofs for quantum key distribution are based on an uncertainty relation \cite{BCCRR10,TR11,TLGR12}. They are also at the heart of security proofs in the bounded quantum storage model \cite{DFSS05, DFRSS07, BFW12}. An uncertainty relation is a statement about a guaranteed uncertainty in the outcome of a measurement in a randomly chosen basis. We refer the reader to \cite{WW10} for a survey on uncertainty relations.

\subsubsection{Uncertainty relation for BB84 measurements}
\label{sec:ur-bb84}
Specifically, here we consider a system $A^n$ of $n$ qubits. Then we measure each one of these qubits in either the standard basis (labeled $0$ with vector $\ket{0}, \ket{1}$) or the Hadamard basis (labeled $1$ with vectors $\ket{+} = {(\ket{0}+\ket{1})}/\sqrt{2}, \ket{-} = {(\ket{0}-\ket{1})}/\sqrt{2}$). More precisely, choose a random vector $\Theta^n \in \{0, 1\}^n$ and measure qubit $i$ in the basis specified by the $i$-th component of $\Theta^n = \Theta_1,\ldots,\Theta_n$. Call the outcome $X_i$. An uncertainty relation is a statement about the amount of uncertainty in the random variable $X^n = X_1,\ldots,X_n$ given the knowledge of the basis choice $\Theta^n$. The uncertainty is often measured in terms of the Shannon entropy. However, for the applications we consider here, the measure of uncertainty needs to be stronger, i.e., we should use a higher order entropy like $\hmin$ or $\htwo$. Such an uncertainty relation has been established in \cite{DFRSS07}:
\begin{equation}
\label{eq:dfrss07}
\entHmin^{\e}(X^n|\Theta^n) \gtrapprox n/2.
\end{equation}
The way this uncertainty relation was used in the context of the bounded storage model was to apply a chain rule to \eqref{eq:dfrss07} to obtain $\entHmin^{\e}(X^n|E\Theta^n) \gtrapprox n/2 - \log |E|$. There are two reasons for this inequality to be unsatisfactory: it depends on the dimension of $E$ rather than on the correlations between $A^n$ and $E$, and it becomes trivial when $\htwo(A^n|E) < -n/2$ as this implies $\log |E| > n/2$. 
%\footnotetext{Here we chose parameters such that $\tau = 0.1$, which yields a lower bound on the \emph{smooth} min-entropy, with smoothing parameter of the order of $10^{-6}$ for a block size of $n=10000$.}

%in general in quantum info but specifically in crypto. qkd and bounded storage. Shannon entropy. Uncertainty relation of a tight high order... What if quantum side information.
%This uncertainty relation is tight but does not take into account a quantum adversary holding a quantum system $E$ that might be entangled with the system $A^n$ being measured.
%Note that this is from the point of view of a classical observer. We could imagine that the observer holds a quantum system $E$ that might be entangled with the system being measured. 
It is simple to see that if the system $A^n$ is maximally entangled with some system $E$, then the outcome $X^n$ of this measurement can be perfectly predicted by having access to $E$. In other words, if the conditional entropy $\htwo(A^n|E) = -n$, then $X^n$ can be correctly guessed with probability $1$. The following theorem provides a converse: if $\htwo(A^n|E) \geqslant -(1-\e) n$ for $\e > 0$, then $X^n$ cannot be guessed with probability better than $2^{-n \delta(\e)}$ with $\delta(\e) > 0$ whenever $\e > 0$. 
%It is worth stressing on the differences with the uncertainty relation in \eqref{eq:dfrss07}. By applying a chain rule to \eqref{eq:dfrss07}, we have $\entHmin^{\e}(X|E\Theta) \gtrapprox n/2 - \log |E|$. 
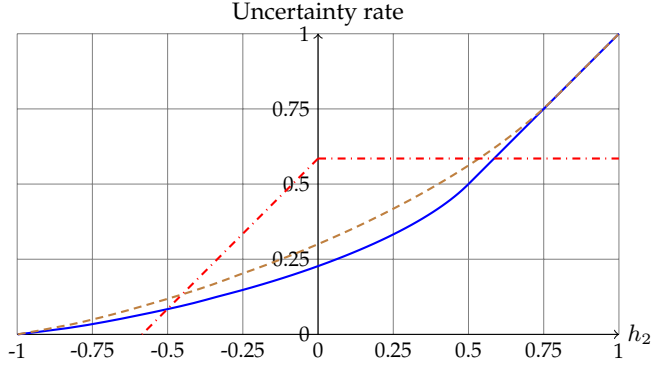
\begin{figure}
	\centering
	\begin{tikzpicture}[scale=4]
		\draw[very thin,color=gray,step=0.25] (-1,0) grid (1,1);
		\draw[->] (-1,0) -- (1,0) node[right] {$h_2$};
		\foreach \x in {-1, -0.75, -0.5, -0.25, 0, 0.25, 0.5, 0.75, 1}
			\node[below,font=\footnotesize] at (\x cm,0) {\x};
		\draw[->] (0,0) -- (0,1) node[above] {Uncertainty rate};
		\foreach \y in {0, 0.25, 0.5, 0.75, 1}
			\node[left,font=\footnotesize] at (0,\y cm) {\y};
		%\draw[blue] plot[smooth] file {bb84.txt};
		%\draw[red] plot[smooth] file {6state.txt};
		\draw[blue,thick] plot[smooth, domain=0.0001:0.5, variable=\t] ({-\t*log2(\t) - (1-\t)*log2(1-\t) + \t - 1}, {\t});    % gamma, first piece
		\draw[blue,thick] plot[smooth, domain=0.5:1, variable=\t] (\t, \t);    % gamma, second piece
		
		\draw[brown,densely dashed,thick] plot[smooth, domain=0.0001:0.5, variable=\t] ({-\t*log2(\t) - (1-\t)*log2(1-\t) + \t*log2(3) - 1}, {\t*log2(3)});
		\draw[brown,densely dashed,thick] plot[smooth, domain=0.793:1, variable=\t] (\t, \t);
		
		\draw[red,dashdotted,thick] plot[smooth, domain=0.0001:0.5850, variable=\t] ({-\t}, {log2(3)-1-\t});
		\draw[red,dashdotted,thick] plot[smooth, domain=0.0001:1, variable=\t] (\t, 0.5850);		
	\end{tikzpicture}
	\caption{Plot of the function $\gamma(h_2)$ (\protect\tikz[baseline=-0.5ex]{\protect\draw[blue,thick] (0,0) -- (0.5,0);}) from Theorem \ref{thm:ur-h2-bb84} giving a lower bound on the uncertainty of the outcome of BB84 measurement as a function of the entropy rate $h_2$ of the state being measured. For measurements in the six-state bases, the uncertainty rate function we obtain in Theorem \ref{thm:ur-h2-mub} is $\gamma_2(h_2)$ (\protect\tikz[baseline=-0.5ex]{\protect\draw[brown,densely dashed,thick] (0,0) -- (0.5,0);}). For comparison, we also plot the uncertainty rate function proved in~\cite{BFW12} (\protect\tikz[baseline=-0.5ex]{\protect\draw[red,dashdotted,thick] (0,0) -- (0.5,0);}).}
	\label{fig:uncertainty-plots}
\end{figure}

\begin{thm}
\label{thm:ur-h2-bb84}
Let $\rho_{A^nE} \in \cS(A^nE)$ where $A^n$ is an $n$-qubit space and define $h_2 = \frac{\entHtwo(A^n|E)_{\rho}}{n}$. 
Then we have
\begin{align*}
\entHtwo(X^n|E\Theta^n)_{\rho} &\geqslant n \gamma(h_2) - 1 
\end{align*}
where $\rho_{X^nE\Theta^n} = \frac{1}{2^n} \sum_{x^n \in \{0,1\}^n, \theta^n \in \{0, 1\}^n} \proj{x^n} \bra{x^n} H^{\theta^n} \rho_{A^nE} H^{\theta^n} \ket{x^n} \otimes \proj{\theta^n}$ is the state obtained when system $A^n$ is measured in the basis defined in the register $\Theta^n$ and the function $\gamma$ is defined by 
\[
\gamma(h_2) = \left\{
\begin{array}{ll}
h_2 & \text{if } h_2 \geqslant 1/2 \\
g^{-1}(h_2) & \text{if } h_2 < 1/2.
\end{array} \right.
\]
with $g(\alpha) = \binent(\alpha) + \alpha - 1$.
\end{thm}
\begin{proof}
	We apply Theorem \ref{thm:general-h2} with $\cM_{A^n\to X^n \Theta^n} = \cN^{\otimes n}$ where $\cN(\rho) = \frac{1}{2} \sum_{x \in \{0,1\}, \theta \in \{0,1\}} \proj{\theta} \otimes \proj{x} \bra{x} H^{\theta} \rho H^{\theta}\ket{x}$. 
%Let $\tilde{\rho}_{A^nE} = (\id_{A^n} \otimes \rho_E^{-1/4}) \rho_{A^nE} (\id_{A^n} \otimes \rho_E^{-1/4})$ so that $\tr[\tilde{\rho}_{A^nE}^2] = 2^{-\entHtwo(A^n|E)}$. 
% but also $\tr[\cM(U_j \tilde{\rho} U_j^{\dg})^2] = 2^{-\entHtwo(X^n|E)_{\rho^j}}$ where
We have
\begin{align*}
2^{-\entHtwo(X^n\Theta^n|E)_{\cM(\rho)|\rho}} &= \tr\left[ \left(\rho_E^{-1/4} (\cN^{\otimes n} \otimes \id)(\rho_{A^nE}) \rho_E^{-1/4} \right)^2 \right] \\
&= \frac{1}{4^n} \sum_{\theta^n \in \{0,1\}^n} \tr\left[ \left(\rho_E^{-1/4} \sum_{x^n \in \{0,1\}^n} \proj{\theta^n} \otimes \proj{x^n} \bra{x^n} H^{\theta^n} \rho_{A^nE} H^{\theta^n} \ket{x^n} \rho_E^{-1/4} \right)^2\right] \\
&= \frac{1}{2^n} 2^{-\entHtwo(X^n|E\Theta^n)_{\rho}},
\end{align*}
where in the last line we used the expression for the entropy conditioned on a classical system (Lemma \ref{lem:condition-cl}).

We then evaluate the state
\begin{align*}
(\cN^{\dagger} \circ \cN \otimes \id)(\Phi) 
%&= \cN^{\dagger}\Big( \left(\proj{00} + \frac{1}{2}\proj{10} + \frac{1}{2} \proj{11} \right) \otimes \proj{0} +  \left(\frac{1}{2}\proj{10} + \frac{1}{2} \proj{11} \right) \otimes (\ket{0}\bra{1} + \ket{0}\bra{1})\\
%&  + \left(\proj{01} + \frac{1}{2}\proj{10} + \frac{1}{2} \proj{11} \right) \otimes \proj{1} \Big) \\
&= \frac{1}{4}\left(\proj{00} + \proj{11} + \proj{++} + \proj{--} \right) \\
&= \frac{1}{4} \left(\Phi_0 + \frac{1}{2}\Phi_1 + \frac{1}{2}\Phi_3\right)\ ,
\end{align*}
%
% The map N^{\dagger}(Y) = \bra{x^n \theta^n} Y \ket{x^n \theta^n} \proj{x^n_{\theta^n}}.
%
%
where $\Phi_i$ are defined in Equations \eqref{eqn:bellbasis01} and \eqref{eqn:bellbasis23}.~\footnote{Note that $\Phi_2$ is the projector on the anti-symmetric subspace and hence cannot appear in this decomposition.} 
In the notation of Theorem \ref{thm:general-h2}, we have for the map $\cM$ and for $s\in \{0,1,3\}^n$, $\lambda_s = \frac{1}{4^n} \cdot \frac{1}{2^{|s|}}$. For $s \notin \{0,1,3\}^n$, $\lambda_s = 0$. As a result, when applying Theorem \ref{thm:general-h2}, it is natural to choose the partition $\mfS_+ \cup \mfS_-$ of the form $\mfS_+ = \{ s \in [d^2]^n : |s| \leqslant \ell_0\}$ and $\mfS_- = \{s \in [d^2]^n : |s| > \ell_0\}$ for a value of $\ell_0 \in \{0, \dots, n\}$ to be chosen as a function of $h_2$. We obtain for any $\ell_0$
\begin{align}
2^{-\entHtwo(X^n|E\Theta^n)_{\rho}}
&\leqslant \sum_{\ell=0}^{\ell_0} \binom{n}{\ell} 2^{-h_2 n - n} + 2^{-\ell_0 -1} \delta_{\ell_0 \leq n-1}\ , \label{eqn:ell0}
%\sum_{s \in \{0,1,3\}^n, |s| = \ell} \tr\left[ \tilde{\rho}_{AE} \otimes \top(\tilde{\rho}_{AE}) \Phi_{s}\otimes \Phi_{E\bar{E}} \right].
\end{align}
where $\delta_{\ell_0 \leq n-1} = 1$ if $\ell_0 \leq n-1$ and $0$ if $\ell_0 = n$. If $h_2 \geqslant 1/2$, let $\ell_0 = n$, in which case we obtain a bound of 
\begin{align*}
2^{-\entHtwo(X^n|E\Theta^n)_{\rho}}
&\leqslant 2^{-h_2 n}.
\end{align*}
If $h_2 < 1/2$, then we are going to choose $\ell_0 \leqslant n/2$. Define the function $g(\alpha) = h(\alpha) + \alpha - 1$ and let $\alpha_0 \leqslant 1/2$ be such that $g(\alpha_0) = h_2$. We then choose $\ell_0 = \floor{\alpha_0 n}$. As a result,
\begin{align*}
\sum_{\ell=0}^{\ell_0} \binom{n}{\ell} 2^{-h_2 n - n} &\leqslant 2^{n (h(\ell_0/n) - h_2 - 1)} \\
&\leqslant 2^{n (h(\alpha_0) - h_2 - 1)} = 2^{n (-\alpha_0 + 1 + h_2  - h_2 - 1)}
= 2^{-\alpha_0 n}, 
\end{align*}
where the first inequality is due to Lemma \ref{lem:sum-binomial}. In addition, we have $2^{-\ell_0-1} \leqslant 2^{-\alpha_0 n}$. Using these bounds in \eqref{eqn:ell0}, we obtain in this case
\begin{align*}
2^{-\entHtwo(X^n|E\Theta^n)_{\rho}}
&\leqslant 2^{-\alpha_0n + 1}.
\end{align*}
Taking the logarithm leads to the desired result.
\end{proof}

The following corollary expresses the uncertainty relation described in Theorem \ref{thm:ur-h2-bb84} in terms of min-entropies, which will be more convenient for the cryptographic applications.
\begin{cor}
\label{cor:ur-hmin-bb84}
Using the same notation as in Theorem \ref{thm:ur-h2-bb84}, we have
\begin{align}
\label{eqn:ur-hmin-h2-bb84}
\entHmin(X^n|E\Theta^n)_{\rho} &\geqslant \frac{1}{2} (n \gamma(h_2) - 1) \\
&\geqslant \frac{1}{2} (n \gamma(h_{\min}) - 1).
\label{eqn:ur-hmin-hmin-bb84}
\end{align}
where $h_{\min} = \frac{\entHmin(A^n|E)_{\rho}}{n}$. Moreover, for any $\e \in (0,1]$, we have
\begin{equation}
\label{eqn:ur-smoothhmin-bb84}
\entHmin^{\e}(X^n|E\Theta^n)_{\rho} \geqslant n \gamma(h_2) - 1 - \log \frac{2}{\e^2}.
\end{equation}
\end{cor}
%Remark: The pre-factor $1/2$ can be avoided at the expense of some extra smoothing. Probably that's worth doing especially for the uncertainty relation \eqref{eqn:ur-hmin-hmin}.
\begin{proof}
To obtain \eqref{eqn:ur-hmin-h2-bb84}, observe that
$
\entHmin(X^n|E\Theta^n)_{\rho} 
\geqslant \frac{1}{2}\entHtwo(X^n|E\Theta^n)_{\rho},
$
using Lemma \ref{lem:hmin-h2}. To replace $h_2$ by $h_{\min}$, we use the corresponding lower bound in Lemma \ref{lem:hmin-h2-fullyquantum}. To obtain \eqref{eqn:ur-smoothhmin-bb84}, we use Lemma \ref{lem:smoothhmin-h2}.
\end{proof}

\subsubsection{Uncertainty relation for measurements in MUBs}
Consider a system $A^n$ of $n$ qudits and consider a full set of $d+1$ mutually unbiased bases (MUBs) in dimension $d$. A set of bases are said to be mutually unbiased if for any pair of vectors $\ket{v}, \ket{w}$ in two different bases, we have $|\braket{v}{w}| = d^{-1/2}$. We then measure each one of these qudits in a randomly chosen basis from this set. More precisely, choose a random vector $\Theta^n \in [d+1]^n$ and measure qudit $i$ in the basis specified by the $i$-th component of $\Theta^n$. Let $U_{\theta^n}$ be the unitary that transforms the basis $\theta^n$ into the standard basis. We prove an uncertainty relation for these measurements in the presence of quantum side information. Previously, uncertainty relations for these measurements taking into account possible quantum side information were proven in \cite{BFW12}. The main improvement here is that the uncertainty lower bound is nontrivial for any $h_2 > -\log d$. Specifically, for entropy rates $h_2 < -(\log(d+1) - 1)$, this theorem provides the first nontrivial uncertainty rates for single-qudit measurements in MUBs. However, when $h_2$ is close to $0$, the bound of \cite{BFW12} is better than the one provided here. See Figure \ref{fig:uncertainty-plots} for a comparison.

\begin{thm}
\label{thm:ur-h2-mub}
Let $\rho_{A^nE} \in \cS(A^nE)$ where $A^n$ is an $n$-qudit space and define $h_2 = \frac{\htwo(A^n|E)_{\rho}}{n}$. 
Then we have
\[
\entHtwo(X^n|E\Theta^n)_{\rho} \geqslant n \gamma_d(h_2) - 1, 
\]
where $\rho_{X^nE\Theta^n} = \frac{1}{(d+1)^n}\sum_{x \in [d]^n, \Theta^n \in [d+1]^n} \proj{x} \bra{x} U_{\theta^n} \rho_{A^nE} U^{\dagger}_{\theta^n} \ket{x} \otimes \proj{\theta^n}$ is the state obtained when system $A^n$ is measured in the basis defined in the register $\Theta^n$ and 
\[
\gamma_d(h_2) = \left\{
\begin{array}{ll}
h_2 & \text{if } h_2 \geqslant \frac{d-1}{d} \log(d+1) \\
f_d^{-1}(h_2) \log(d+1) & \text{if } h_2 < \frac{d-1}{d} \log(d+1)
\end{array}
\right.
\]
with $f_d(\alpha) = \binent(\alpha) + \alpha \log(d^2-1) - \log d$ defined as in Theorem \ref{thm:h2-sampling}.
\end{thm}
\begin{proof}
	We apply Theorem \ref{thm:general-h2} with $\cM_{A^n \to X^n\Theta^n} = \cN^{\otimes n}$ where $\cN(\rho) = \frac{1}{d+1} \sum_{x \in [d], \theta \in [d+1]} \proj{\theta} \otimes \proj{x} \bra{x} U_{\theta} \rho U^{\dagger}_{\theta} \ket{x}$. Analogous to the proof of Theorem~\ref{thm:ur-h2-bb84}, it is 
	simple to see that $2^{-\entHtwo(X^n\Theta^n|E)_{\cM(\rho)|\rho}} = \frac{1}{(d+1)^n}2^{-\entHtwo(X^n|E\Theta^n)_{\rho}}$. We have in this case
\begin{align*}
((\cN^{\dagger} \circ \cN) \otimes \id)(\Phi) 
&= \frac{1}{(d+1)^2} \sum_{\theta \in [d+1], x \in [d], i,j \in [d]} \bra{x} U_{\theta} \ket{i} \bra{j} U^{\dagger}_{\theta} \ket{x} \: U_{\theta}^{\dagger} \proj{x} U_{\theta} \otimes \ket{i} \bra{j} \\
&= \frac{1}{(d+1)^2} \sum_{\theta \in [d+1], x \in [d]}  U_{\theta}^{\dagger} \proj{x} U_{\theta} \otimes \sum_{i,j \in [d]} \bra{x} U_{\theta} \ket{i} \bra{j} U^{\dagger}_{\theta} \ket{x} \, \ket{i} \bra{j} \\
&= \frac{1}{(d+1)^2} \sum_{\theta \in [d+1], x \in [d]}  U_{\theta}^{\dagger} \proj{x} U_{\theta} \otimes \top(U_{\theta}^{\dagger} \proj{x} U_{\theta}),
\end{align*}
where $\Phi$ is the unnormalized maximally entangled state across two qudits, and $\top$ denotes the transpose with respect to the standard basis.
Now we use the fact that the states $\{U_{\theta} \ket{x} \}_{\theta \in [d+1], x \in [d]}$ form a state two-design \cite{KR05}: 
\[
\sum_{\theta \in [d+1], x \in [d]}  U_{\theta}^{\dagger} \proj{x} U_{\theta} \otimes U_{\theta}^{\dagger} \proj{x} U_{\theta} = \id_{A\bar{A}} + F_{A\bar{A}},
\]
where $F_{A\bar{A}}$ denotes the swap operator $F = \sum_{x,x' \in [d]} \ket{x}\bra{x'} \otimes \ket{x'} \bra{x}$. As $(\ident \otimes \top)(F) = \Phi$, we have 
\[
(\cN^{\dagger} \circ \cN \otimes \ident)(\Phi) = \frac{1}{d+1} \cdot \frac{\id + \Phi}{d+1} = \frac{1}{d+1} \left( \frac{\Phi_0}{d+1} + \frac{\sum_{s \in [d^2]} \Phi_s}{d(d+1)} \right) = \frac{1}{d(d+1)}\left(\Phi_0 + \frac{\sum_{s \neq 0} \Phi_s}{d+1} \right).
\]
This means that for the $n$-fold tensor product $\cM = \cN^{\otimes n}$, we have using the notation of Theorem \ref{thm:general-h2} that $\lambda_s = \frac{1}{(d+1)^n d^n} \frac{1}{(d+1)^{|s|}}$ for all $s \in [d^2]^n$. As a result, when applying Theorem \ref{thm:general-h2}, it is natural to choose the partition $\mfS_+ \cup \mfS_-$ of the form $\mfS_+ = \{ s \in [d^2]^n : |s| \leqslant \ell_0\}$ and $\mfS_- = \{s \in [d^2]^n : |s| > \ell_0\}$ for a value of $\ell_0 \in \{0, \dots, n\}$ to be chosen as a function of $h_2$. We obtain for any $\ell_0$,
\begin{align}
2^{-\entHtwo(X^n|E\Theta^n)_{\rho}}
&\leqslant \sum_{\ell=0}^{\ell_0} \binom{n}{\ell} (d^2-1)^{\ell} 2^{-h_2 n} (d+1)^{-\ell} d^{-n} + (d+1)^{-\ell_0 -1} \delta_{\ell_0 \leqslant n-1} \notag \\
&= \sum_{\ell=0}^{\ell_0} \binom{n}{\ell} (d-1)^{\ell} 2^{-n h_2  - n \log d} + (d+1)^{-\ell_0 -1} \delta_{\ell_0 \leqslant n-1},\label{eqn:ell0-mub}
%\sum_{s \in \{0,1,3\}^n, |s| = \ell} \tr\left[ \tilde{\rho}_{AE} \otimes \top(\tilde{\rho}_{AE}) \Phi_{s}\otimes \Phi_{E\bar{E}} \right].
\end{align}
where $\delta_{\ell_0 \leqslant n-1} = 1$ if $\ell_0 \leqslant n-1$ and $0$ otherwise.
If $h_2 \geqslant \frac{d-1}{d} \log(d+1)$, let $\ell_0 = n$, in which case we obtain a bound of 
\begin{align*}
2^{-\entHtwo(X^n|E\Theta^n)_{\rho}}
&\leqslant \sum_{\ell=0}^{n} \binom{n}{\ell} (d-1)^{\ell} 2^{-n h_2  - n \log d} = 2^{-h_2 n}.
\end{align*}
If $h_2 < \frac{d-1}{d} \log(d+1)$, then we are going to choose $\ell_0 \leqslant \frac{d-1}{d}n$. Note that $f_d(\frac{d-1}{d}) = \frac{d-1}{d} \log(d+1)$. As $h_2 < \frac{d-1}{d} \log(d+1)$ and $f_d$ is nondecreasing on $[0, (d-1)/d]$, we can find $\alpha_0 \leqslant (d-1)/d$ be such that $f_d(\alpha_0) = h_2$. We then choose $\ell_0 = \floor{\alpha_0 n}$. As a result,
\begin{align*}
\sum_{\ell=0}^{\ell_0} \binom{n}{\ell} (d-1)^{\ell} 2^{-h_2 n - (\log d) n} &\leqslant 2^{n (h(\ell_0/n) + \ell_0/n \log (d-1)) - n (h_2 + \log d))} \\
&\leqslant 2^{n (h(\alpha_0) + \alpha_0 \log(d-1) - h_2 - \log d)} \\
&= 2^{n (-\alpha_0 \log(d+1) + \log d + h_2  - h_2 - \log d)}
= (d+1)^{-\alpha_0 n}, 
\end{align*}
where the first inequality is due to Lemma \ref{lem:sum-binomial}. In addition, we have $(d+1)^{-\ell_0-1} \leqslant (d+1)^{-\alpha_0 n}$. Using these bounds in \eqref{eqn:ell0}, we obtain in this case
\begin{align*}
2^{-\entHtwo(X^n|E\Theta^n)_{\rho}}
&\leqslant 2 (d+1)^{-\alpha_0n}.
\end{align*}
Taking the logarithm leads to the desired result.
\end{proof}

%The min-entropy has a more appealing operational interpretation and it is tightly related to the two-entropy. 
The following corollary expresses the uncertainty relation described in Theorem \ref{thm:ur-h2-mub} in terms of min-entropies. The proof is the same as Corollary \ref{cor:ur-hmin-bb84}. 
\begin{cor}
\label{cor:ur-hmin}
Using the same notation as in Theorem \ref{thm:ur-h2-mub}, we have
\begin{align}
\label{eqn:ur-hmin-h2}
\entHmin(X^n|E\Theta^n)_{\rho} &\geqslant \frac{1}{2} \left( n \gamma_d(h_2)-1 \right) \\
\label{eqn:ur-hmin-hmin}
&\geqslant \frac{1}{2}\left(n \gamma_d(h_{\min})-1\right).
\end{align}
where $h_{\min} = \frac{\entHmin(A|E)_{\rho}}{n}$. Moreover, for any $\e \in (0,1]$, we have
\begin{equation}
\label{eqn:ur-smoothhmin-mubs}
\entHmin^{\e}(X^n|E\Theta^n)_{\rho} \geqslant n \gamma_d(h_2) - 1 - \log \frac{2}{\e^2}.
\end{equation}
\end{cor}

\subsection{Security in the noisy-storage model}
\label{sec:nsm}
\subsubsection{General noisy storage model}

We now use our new uncertainty relations to prove that the primitive weak string erasure can be secure as soon as one of the parties has a memory that cannot reliably store $n$ qubits. 
In weak string erasure, the objective is to generate a string $X^n$ such that Alice holds $X^n$ and Bob holds a random subset $I \subseteq [n]$ and the bits $X_I$ of $X^n$ corresponding to the indices in $I$. Randomly chosen here means that each index $i \in [n]$ has probability $1/2$ of being in $I$. The security criterion is that at the end of the protocol, a cheating Bob should have a state satisfying $\entHmin(X^n|B) \geqslant \lambda n$ where $B$ represents Bob's system, and a cheating Alice should not learn anything about $I$. To summarize all relevant parameters, we speak of an $(n,\lambda)$-WSE scheme and refer to~\cite{KWW09} for a definition~\footnote{Note that the original definition includes a security error $\eps$, which in our case is $\eps = 0$.}. It is proved in~\cite{KWW09} that bit commitment can be implemented using weak string erasure and classical communication.

%\omar{I'm not sure what the $\e$ is exactly, do we need it here? Because we get a bound on a non-smoothed entropy.}\fred{I agree; I removed the $\eps$ from the definition.}\steph{I have added a comment that we do not put $\varepsilon$}

\textbf{Protocol.} The protocol we use here is the same as the one of \cite{KWW09}. Alice prepares a random string $X^n \in \{0,1\}^n$ and encodes each bit $X_i$ in either the standard basis $\Theta_i = 0$ or the Hadamard basis $\Theta_i=1$, each with probability $1/2$. Then Bob measures these qubits in randomly chosen bases $\Theta'_i$. After the waiting time, Alice reveals both $X^n$ and $\Theta^n$. The set $I$ is defined by $I = \{i : \Theta_i = \Theta'_i\}$. For a more detailed description of the protocol, we refer the reader to \cite{KWW09}.

To state the result, we first define the notion of \emph{channel fidelity} introduced by \cite{BKN98} which is perhaps the most widely used quantity to measure how good a channel is at sending quantum information. For a channel $\cN : \cS(Q) \to \cS(Q')$, the channel fidelity $F_c$ quantifies how well $\cN$ preserves entanglement with a reference:
\begin{equation}
\label{eq:channel-fidelity}
F_c(\cN) = F( \Phi^N_{Q'A}, \left[\cN \otimes \id_{A} \right] (\Phi^N_{QA}) ),
\end{equation}
where $\Phi^N_{QA}$ is a normalized maximally entangled state.
For example, one way of defining the (one-shot) quantum capacity with free classical forward communication of a channel $\cF_{B \rightarrow C}$ is by the maximum of $\log |Q|$ over all encodings $\cE : \cS(Q) \to \cS(B \otimes M)$ and decodings $\cD : \cS(C \otimes M) \to \cS(Q')$ such that $F_c(\cD \circ (\cF \otimes \clchannel_M) \circ \cE) \geqslant 1-\e$ for small enough $\e$. Here $\clchannel_M$ refers to a noiseless classical channel.

The following theorem states that as soon as the storage device of Bob cannot send quantum information with reliability better than $\eta$, then we can perform two-party computation securely provided $\eta \leqslant 2^{-\delta n}$ for any $\delta > 0$. One can even obtain security when $\eta \leqslant 2^{-c(\log^2n + \log n \log(1/\e))}$ for some large enough constant $c$. Previously, this was only known when $\eta < 2^{-(2-\log 3) n}$ \cite{BFW12}. 

\begin{thm}\label{thm:nsm}
	Let Bob's storage device be given by $\cF : \cS(\hin) \to \cS(B)$, and let $\eta \in (0,1)$. Assume that we have
	\begin{equation}
	\label{eqn:fid-condition}
	\max_{\cD,\cE} F_c(\cD \circ (\cF \otimes \clchannel_{\msg}) \circ \cE)^2 \leqslant \eta %2^{-\lambda n - \kappa}
	\end{equation}
	where the maximum is over all quantum channels $\cE : \cS\left((\mbC^2)^{\otimes n}\right) \to \cS(\hin \otimes M)$ and $\cD:\cS(B \otimes M) \to \cS((\mbC^2)^{\otimes n})$.
	
	Then, the protocol described above implements a $(n, \lambda)$-WSE for 
	\[
	\lambda = \frac{1}{2} \left( \gamma\left(-1 + \log(1/\eta)/n\right) - \frac{1}{n} \right).
	\]
\end{thm}
\begin{proof}
	The proof of correctness of the protocol, and security against dishonest Alice is identical to~\cite{KWW09} and does not lead to
	any error terms. 

For the security against dishonest Bob, it is convenient to imagine a purification of the protocol, in which Alice prepares $n$ EPR pairs $\Phi^N_{A^nQ}$, where she sends $Q$ to Bob
and later measures her $n$ qubits $A^n$ in randomly chosen BB84 bases.
Bob's general attack is illustrated in Figure \ref{fig:BobAttack}. We use the uncertainty relation in Equation \eqref{eqn:ur-hmin-hmin-bb84}, with $E=BM \Theta^n$ on ${\rho}_{A^nBM\Theta^n}$. In order to do that, we first derive a lower bound on $h_{\min} = \frac{\entHmin(A^n|BM\Theta^n)_{\rho}}{n}$. Note that because $\Theta^n$ is independent of $A^nBM$, we have $\entHmin(A^n|BM\Theta^n)_{\rho} = \entHmin(A^n|BM)_{\rho}$. We now use Condition \eqref{eqn:fid-condition} to obtain a lower bound on $\entHmin(A^n|BM)$.

\begin{figure}[h]
\begin{center}
		\includegraphics[scale=0.7]{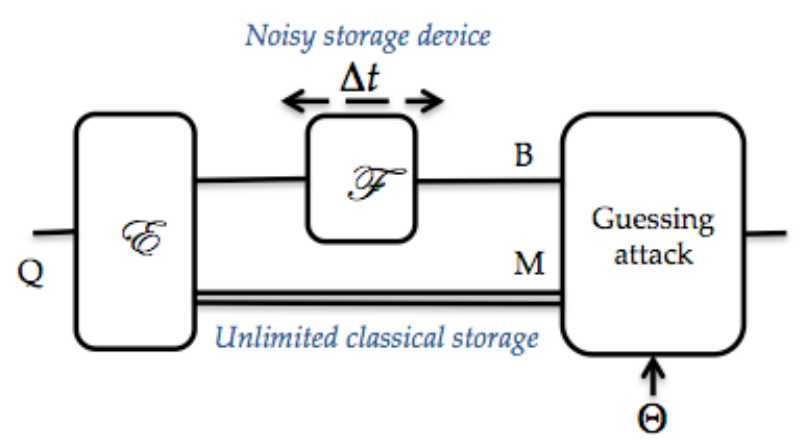}
		\caption{An attack of dishonest Bob is described by an encoding attack $\cE$ and a guessing attack because for classical $X^n$ the min-entropy $\hmin(X^n|BM\Theta^n)$ is directly related to the probability that Bob guesses $X^n$. The uncertainty relation of \eqref{eqn:ur-hmin-hmin} is going to allow us to relate this guessing probability to how well a decoding attack $\cD$ can preserve entanglement between Alice and Bob, where $\cD$ acts on $BM$.}
		\label{fig:BobAttack}
\end{center}
\end{figure}

In fact, we use an operational interpretation of the conditional min-entropy due to~\cite{KRS09}:
\begin{align}\label{eqn:hmin-oper}
	\entHmin(A^n|BM)_{{\rho}} = - \log |A^n| \max_{\Lambda_{BM\rightarrow \bar{A}^n}} F(\Phi^N_{A^n\bar{A}^n},\id_{A^n} \otimes \Lambda(\rho_{A^nBM}))^2\ ,
\end{align}
where $\Phi_{A^n\bar{A}^n}^N$ is the normalized maximally entangled state across $A^n\bar{A}^n$. 
That is, the min-entropy is directly related to the ``amount'' of entanglement between $A^n$ and $BM$. 
The map $\Lambda$ in~\eqref{eqn:hmin-oper} can be understood as a decoding attack $\cD$ aiming to restore entanglement with Alice.

Further, note that the expression in \eqref{eqn:hmin-oper} is the same as
\begin{align}\label{eq:chanFid}
	\max_{\cD,\cE} F\left(\Phi^N_{A^nB},\id_{A^n} \otimes \left[\cD \circ (\cF \otimes \clchannel_{\msg}) \circ \cE\right](\Phi^N_{A^nQ})\right) = 
	\max_{\cD,\cE} F_c(\cD \circ (\cF \otimes \clchannel_{\msg}) \circ \cE)\ .
\end{align}
	By the assumption on the storage device $\cF$, we obtain that for any encoding $\cE$ and decoding $\cD$ attack of Bob
	\begin{align*}
		\entHmin(A^n|BM)_{\rho}  &\geqslant - \log 2^n F_c(\cD \circ (\cF \otimes \clchannel_{\msg}) \circ \cE)^2\\
						&\geqslant - \left(n - \log(1/\eta) \right). 
	\end{align*}
	Then, using the uncertainty relation of \eqref{eqn:ur-hmin-hmin-bb84}, we obtain
	\begin{align*}%\label{eq:entBound}
		\entHmin(X^n|BM\Theta^n)_{{\rho}} \geq
		\frac{1}{2} \left( n \gamma\left(-1 + \log(1/\eta)/n\right) -1 \right),
		%\left(n \log \left(\frac{3}{1 + 2^{-R_2\left(-1 + \frac{\log(1/\eta)}{n}\right)}}\right) - 2\log n - 1\right),
	\end{align*}
	which proves the desired result.
\end{proof}

\subsubsection{Special case: bounded storage model}
%In this case, we can directly use the uncertainty relation in \eqref{eqn:ur-hmin-h2} by bounding the collision entropy rate $h_2$ as a function of that available memory. This allows us to obtain better parameters than in the general noisy storage model case.
The next theorem simply states the result in the important special case of the bounded storage model.
\begin{thm}[WSE in the bounded storage model]
\label{thm:bqsm}
If Alice has $q$ qubits of quantum memory then the protocol described in the previous section implements $(n, \lambda)$-WSE with $\lambda = \frac{1}{2}\left(\gamma(-q/n)-\frac{1}{n}\right)$.
\end{thm}
\begin{proof}
The proof is the same as Theorem \ref{thm:nsm}, but we can now directly obtain a lower bound on $\entHtwo(A^n|BM)_{\rho} \geqslant -q$ using Lemma \ref{lem:condition-cl}. By \eqref{eqn:ur-hmin-h2-bb84}, we have 
	\begin{align*}
		\entHmin(X^n|BM\Theta^n)_{{\rho}} \geq
		\frac{1}{2} (n \gamma(-q/n) - 1).
	\end{align*}
\end{proof}
Previously, in this case, security was only proven when $q < \frac{2n}{3}$ \cite{MW10} with a variant of this protocol that uses a six-state encoding. Using the estimate in Claim \ref{claim:estimate-gamma}, the previous theorem shows that $q < n - c\log^2 n$ for some large enough $c$ would be sufficient to perform WSE securely. Using the construction of \cite{KWW09}, this leads to a secure bit commitment provided $q < n - c \log^2 n - c \log n \log(1/\e)$ for some large enough constant $c$ and where $\e$ is the failure probability.

\section{Conclusion}
\label{sec:conclusion}
We have determined a bound on how the min-entropy changes when $A$ is transformed to $\mathcal{M}(A)$ for a certain general class of processes $\mathcal{M}$.  Our results on entanglement sampling, as well as uncertainty relations with respect to quantum side information then follow naturally for different choices of $\mathcal{M}$. Our results on entanglement sampling have in fact already found applications in the context of studying properties of random quantum circuits \cite{BF13}. 

One important aspect of our results compared to previous works on uncertainty relations and quantum random access codes is to give nontrivial bounds for all the range of possible min-entropy of the input. However, for some specific ranges of the input entropy, other techniques lead to better rates. It would be interesting to see if it is possible to combine our techniques with ideas from previous work such as \cite{BFW12} for uncertainty relations or \cite{BARdW08} for random access codes to obtain tight bounds.
It is likely that other interesting statements can be made using Theorem \ref{thm:general-h2} for different maps, and it is an interesting open question to extend our results to more general maps.

\acknowledgments
We thank Oleg Szehr, Marco Tomamichel and Thomas Vidick for useful discussions. We would also like to thank Serge Fehr, Maximilian Fillinger, Frédéric Grosshans and Christian Schaffner for pointing out some inaccuracies in a previous version and for helping us improve the presentation.
FD acknowledges support from the Danish National Research Foundation and The National Science Foundation of China (under the grant 61061130540) for the Sino-Danish Center for the Theory of Interactive Computation, within which part of this work was performed; and also from the CFEM research center (supported by the Danish Strategic Research Council) within which part of this work was performed. 
OF is supported by the European Research Council grant No. 258932.
SW is supported by the National Research Foundation and the Ministry of Education, Singapore. 
SW thanks ETH Z\"{u}rich for their hospitality.

\bibliography{ur}

\appendix

\section{Technical Lemmas}

\RestateBinomialSum*
\begin{proof}
It is convenient to study separately the case where $\ell_0 \leqslant \frac{d^2-1}{d^2} (n-k)$ and the case where $\ell_0 > \frac{d^2-1}{d^2} (n-k)$. More precisely, the following claim introduces the value $k_0$ that separates these two cases.
%We will do so by considering two cases $k \leqslant k_0$ and $k > k_0$. In order to define $k_0$ observe first that $\sum_{0 \leqslant \ell \leqslant \ell_0} {n-k \choose \ell} (d^2-1)^{\ell} \leqslant (d^2)^{n-k}$. Note that the maximum of $\binom{n-k}{\ell} (d^2-1)^{\ell}$ is achieved for $\ell \approx \frac{d^2-1}{d^2} (n-k)$, so $k_0$ is chosen so that for $k \leqslant k_0$, $\ell_0 \leqslant \frac{d^2-1}{d^2} (n-k)$. More precisely, we now define $k_0$ as follows.
\begin{claim}
\label{claim:k0}
There exists $k_0 \in \{1, \dots, n\}$ such that $\ell_0 \leqslant \frac{d^2-1}{d^2} (n-k_0+1)$ such that 
\begin{enumerate}
\item for $k \leqslant k_0$, $\sum_{\ell \leqslant \ell_0} {n-k \choose \ell} (d^2-1)^{\ell} \leqslant n \cdot {n-k \choose \ell_0} (d^2-1)^{\ell_0}$,
\item and $\sum_{\ell \leqslant n-k_0} {n-k_0 \choose \ell} (d^2-1)^{\ell} = (d^2)^{n-k_0} \leqslant n \cdot {n-k_0 \choose \ell_0} (d^2-1)^{\ell_0}$.
\end{enumerate}
\end{claim}
\begin{proof}
We have for $\ell \geqslant 1$,
\begin{align*}
\frac{{n-k \choose {\ell}} (d^2-1)^{\ell}}{{n-k \choose {\ell-1}} (d^2-1)^{\ell-1}} 
&= (d^2-1) \frac{n-k-\ell+1}{\ell}.
\end{align*}
Now define $\ell_{\max}(k)$ to be the largest integer such that $\ell_{\max}(k) \leqslant \frac{d^2-1}{d^2}(n-k+1)$. In this case, we have for $\ell \leqslant \ell_{\max}(k)$, $(d^2-1) \frac{n-k-\ell+1}{\ell} \geqslant 1$. 
As a result, we have that for a fixed $k$, the expression $\binom{n-k}{\ell} \left(d^2-1\right)^\ell$ is increasing for $\ell \leqslant \ell_{\max}(k)$. In addition, if $\ell > \ell_{\max}(k)$, then $\ell > \frac{d^2-1}{d^2}(n-k+1)$ which means that for $\ell > \ell_{\max}(k)$, the expression ${n-k \choose {\ell}} (d^2-1)^{\ell}$ is decreasing. 

We choose $k_0$ to be the largest integer such that $\ell_0 \leqslant \frac{d^2-1}{d^2} (n-k_0+1)$. Note that such a $k_0$ exists because we assumed $\ell_0 \leqslant \frac{d^2-1}{d^2} n$. This means that $\ell_0 > \frac{d^2-1}{d^2} (n - k_0) \geqslant \frac{d^2-1}{d^2} (n - k_0+1) - 1$. This implies that $\ell_0 = \ell_{\max}(k_0)$ is the maximum of $\binom{n-k_0}{\ell} \left(d^2-1\right)^\ell$ over $\ell \in \{0, \dots, n-k_0\}$. Both points then follows from bounding the sum by $n$ times the largest term. %The first point follows from the fact that for $k \leqslant k_0$, then $\ell_{\max}(k) \geqslant \ell_{\max}(k_0) = w$. The second point follows from the fact that 
%\[
%\binom{n-k_0}{\ell_{\max}(k_0)} \left(d^2-1\right)^{\ell_{\max}(k_0)} \geqslant \frac{(d^2)^{n-k_0}}{n-k_0 + 1} \geqslant \frac{(d^2)^{n-k_0}}{n}.
%\]
%This concludes the proof of the claim.
\end{proof}
As a result, we have for $k \leqslant k_0$, 
\begin{align}
\sum_{\ell=0}^{\ell_0} \binom{n-k}{\ell} (d^2-1)^{\ell} &\leqslant n \cdot {n-k \choose \ell_0} (d^2-1)^{\ell_0} \notag \\
&= n \cdot {n \choose \ell_0} (d^2-1)^{\ell_0} \frac{(n-\ell_0) \cdots (n-\ell_0-k+1)}{n \cdots (n-k+1)}. \notag
\end{align}
%If $k = 1$, then the bound gives $n \cdot {n \choose \ell_0} (d^2-1)^{\ell_0} \frac{n-\ell_0}{n} \leqslant n^2 \cdot {n \choose \ell_0} (d^2-1)^{\ell_0} \frac{1}{n}$. 
Note that for $k = 1$, the result simply follows from the fact that $n-\ell_0-1 \geqslant 1$, which itself comes from our assumptions $\ell_0 \leqslant \frac{d^2-1}{d^2} n$ and $d^2 < n$. For $k > 1$, we can continue with
%Note first that we have $k \leqslant k_0 \leqslant n-\ell_0$. 
%In the case $n-\ell_0 = k$, we obtain a bound of 
\begin{align*}
\sum_{\ell=0}^{\ell_0} \binom{n-k}{\ell} (d^2-1)^{\ell} 
&\leqslant n \cdot {n \choose \ell_0} (d^2-1)^{\ell_0} (n-\ell_0) \frac{(n-\ell_0-1) \cdots (n-\ell_0-k+1)}{n \cdots (n-k+1)} \\
&\leqslant n^2 \cdot {n \choose \ell_0} (d^2-1)^{\ell_0} \left(\frac{n-\ell_0-1}{n} \right)^k.
\end{align*}
%Note that the second equality above only holds when $n-\ell_0-k \neq 0$. However, if $n-k = \ell_0$, we have $\frac{1}{\binom{n}{k}} = \frac{k!}{n (n-1) \cdots (n-k+1)} \leqslant k \left(\frac{\max(n-\ell_0-1,1)}{n}\right)^k$ so inequality \eqref{eqn:sum-smaller-w-2} still holds in this case.
For $k > k_0$,
\begin{align}
\sum_{\ell=0}^{\ell_0} \binom{n-k}{\ell} (d^2-1)^{\ell} &\leqslant d^{2(n-k)} \notag \\
&\leqslant n \cdot {n-k_0 \choose \ell_0} (d^2-1)^{\ell_0} d^{-2(k-k_0)} \notag \\
&= n \cdot {n \choose \ell_0} (d^2-1)^{\ell_0} \frac{(n-\ell_0) \cdots (n-\ell_0-k_0+1)}{n \cdots (n-k_0+1)} \left(\frac{1}{d^2}\right)^{k-k_0}. \label{eqn:bound-k_0-d2}
\end{align}
%Again for $k_0 = 1$, the result simply follows from the fact that $n-\ell_0-1 \geqslant 1$. 
For $k_0 > 1$, we use the fact that $\ell_0 \leqslant \frac{d^2-1}{d^2} (n-k_0+1)$, which implies that $\frac{1}{d^2} \leqslant \frac{n-\ell_0-k_0+1}{n-k_0+1}$. Thus,
%In the case $n-\ell_0 = k_0$, we obtain a bound of $n \cdot {n \choose \ell_0} (d^2-1)^{\ell_0} \frac{1}{\binom{n}{k_0} d^{2(k-k_0)}}$ which proves the desired result.
\begin{align*}
\sum_{\ell=0}^{\ell_0} \binom{n-k}{\ell} (d^2-1)^{\ell} &\leqslant d^{2(n-k)} \\
&\leqslant n \cdot {n \choose \ell_0} (d^2-1)^{\ell_0} \frac{(n-\ell_0) \cdots (n-\ell_0-k_0+1)^{(k-k_0)+1}}{n \cdots (n-k_0+1)^{(k-k_0)+1}} \\
&\leqslant n \cdot {n \choose \ell_0} (d^2-1)^{\ell_0} (n-\ell_0) \left(\frac{n-\ell_0-1}{n} \right)^{k} \\
&\leqslant n^2 \cdot {n \choose \ell_0} (d^2-1)^{\ell_0} \left(\frac{n-\ell_0-1}{n} \right)^{k}.
\end{align*}
For $k_0 = 1$, \eqref{eqn:bound-k_0-d2} becomes
\begin{align*}
\sum_{\ell=0}^{\ell_0} \binom{n-k}{\ell} (d^2-1)^{\ell} 
&\leq n \cdot {n \choose \ell_0} (d^2-1)^{\ell_0} \frac{n-\ell_0}{n} \left(\frac{1}{d^2}\right)^{k-1} \\
&\leq n^2 {n \choose \ell_0} (d^2-1)^{\ell_0} \left(\frac{1}{d^2}\right)^{k} \ ,
\end{align*}
using the assumption $n > d^2$.
%\begin{align*}
%\sum_{\ell=0}^{\ell_0} \binom{n-k}{\ell} (d^2-1)^{\ell} &\leqslant d^{2(n-k)}
%&\leqslant n \cdot {n \choose \ell_0} (d^2-1)^{\ell_0} (n-\ell_0) \max\left(\left(\frac{(n-\ell_0-1)}{n}\right)^k, \left(\frac{1}{d^2}\right)^k \right) \\
%&\leqslant n^2 \cdot {n \choose \ell_0} (d^2-1)^{\ell_0} (n-\ell_0) \left(\frac{n-\ell_0-1}{n} \right)^{k} \\
%&\leqslant n^2 \cdot {n \choose \ell_0} (d^2-1)^{\ell_0} \left(\frac{n-\ell_0-1}{n} \right)^{k}.
%\end{align*}
%\omar{There is problem if $k_0 = 1$ here, I'm not sure how to get this same bound, we can replace it by $(n-\ell_0)/n (d^2)^{-k+1}$, but I'm not sure about the claimed $1/(d^2)^k$ claimed in the statement. }
\end{proof}

\section{Some useful properties of entropy measures}

\begin{lem}\label{lem:hmin-h2-fullyquantum}
	Let $\rho_{AB} \in \cS_{\leqslant}(AB)$. Then, $\hmin(A|B)_{\rho} \leqslant \htwo(A|B)_{\rho} \leqslant 2\hmin(A|B)_{\rho} + \log d_A$.
\end{lem}
\begin{proof}
	The first inequality can be proven as follows:
	\begin{align*}
		2^{-\hmin(A|B)_{\rho}} &= \max_{E_{AB}:E_B = \ident_B} \tr[E_{AB} \rho_{AB}]\\
		&\geqslant \tr[(\rho_B^{-1/2} \rho_{AB} \rho_B^{-1/2}) \rho_{AB}]\\
		&= 2^{-\htwo(A|B)_{\rho}}.
	\end{align*}
	For the second inequality, we proceed as follows. By \cite{KRS09}, there exists a CPTP map $\mathcal{E}_{B \rightarrow A'}$ with $A' \cong A$, such that $\hmin(A|B)_{\rho} = \hmin(A|A')_{\mathcal{E}(\rho)}$. Letting $\tilde{\rho} = \mathcal{E}(\rho)$ and $\omega_{A'} = \sqrt{\tilde{\rho}_{A'}}/\tr[\sqrt{\tilde{\rho}_{A'}}]$, we get
	\begin{align*}
		2^{-\hmin(A|B)_{\rho}} &= 2^{-\hmin(A|A')_{\tilde{\rho}}}\\
		&\leqslant 2^{-\hmin(A|A')_{\tilde{\rho}|\omega}}\\
		&= \left\| \omega_{A'}^{-1/2} \tilde{\rho}_{AA'} \omega_{A'}^{-1/2} \right\|_{\infty}\\
		&= \left\| \tilde{\rho}_{A'}^{-1/4} \tilde{\rho}_{AA'} \tilde{\rho}_{A'}^{-1/4} \right\|_{\infty} \tr[\sqrt{\tilde{\rho}_{A'}}]\\
		&\leqslant \left\| \tilde{\rho}_{A'}^{-1/4} \tilde{\rho}_{AA'} \tilde{\rho}_{A'}^{-1/4} \right\|_{2} \tr[\sqrt{\tilde{\rho}_{A'}}]\\
		&= \sqrt{2^{-\htwo(A|A')_{\tilde{\rho}|\tilde{\rho}}}} \tr[\sqrt{\tilde{\rho}_{A'}}]\\
		&\leqslant \sqrt{2^{-\htwo(A|B)_{\rho}}} \tr[\sqrt{\tilde{\rho}_{A'}}]\\
		&\leqslant \sqrt{d_A 2^{-\htwo(A|B)_{\rho}}},
	\end{align*}
	and the lemma follows.
\end{proof}

\begin{lem}\label{lem:hmin-h2}
	Let $\rho_{XB}\in\cS_{\leqslant}(XB)$ be a CQ state. %, and let $\sigma \in \cS_{\leqslant}(B)$ be such that $\hmin(X|B)_{\rho} = \hmin(X|B)_{\rho|\sigma}$. 
	Then
\begin{align*}
	\entHmin(X|B)_{\rho} \leqslant \entHtwo(X|B)_{\rho} \leqslant 2 \entHmin(X|B)_{\rho} .
\end{align*}
\end{lem}
\begin{proof} The lower bound is a special case of Lemma \ref{lem:hmin-h2-fullyquantum}. For the upper bound,
	from the operational interpretation of $\hmin$, we get that there exists a measurement $\mathcal{M}_{B \rightarrow X'}$ such that $\hmin(X|B)_{\rho} = \hmin(X|X')_{\mathcal{M}(\rho)}$. Using this, we get that
	\begin{align*}
		2^{-\hmin(X|B)_{\rho}} &= 2^{-\hmin(X|X')_{\mathcal{M}(\rho)}}\\
		&= \mbE_{\bar{X}'} 2^{-\hmin(X|X'=\bar{X}')_{\mathcal{M}(\rho)}}\\
		&\leqslant \mbE_{\bar{X}'} 2^{-\demi \htwo(X|X'=\bar{X}')_{\mathcal{M}(\rho)}}\\
		&\leqslant \sqrt{\mbE_{\bar{X}'} 2^{-\htwo(X|X'=\bar{X}')_{\mathcal{M}(\rho)}}}\\
		&= \sqrt{2^{-\htwo(X|X')_{\mathcal{M}(\rho)|\mathcal{M}(\rho)}}}\\
		&\leqslant \sqrt{2^{-\htwo(X|B)_{\rho}}}.
	\end{align*}
	where the first inequality follows from an application of Cauchy-Schwarz, the second from the concavity of the square root, and the third from the monotonicity of $\htwo$ under CPTP maps. The last equality is due to the following:
\begin{align*}
\mbE_{\bar{X}'} 2^{-\htwo(X|X'=\bar{X}')_{\mathcal{M}(\rho)}} &= \sum_{x'} \pr{X'=x'} \sum_{x} \pr{X=x|X'=x'}^2 \\
&= \tr[\left((\id_{X} \otimes \rho_{X'}^{-1/4}) \rho_{XX'} (\id_X \otimes \rho_{X'}^{-1/4}) \right)^2] \\
&= 2^{-\entHtwo(X|X')_{\cM(\rho)|\cM(\rho)}}.
\end{align*} 
\end{proof}

\begin{lem}\label{lem:smoothhmin-h2}
	Let $\rho_{AB} \in \cS_{\leqslant}(AB)$ and $\sigma_B \in \cS_{\leq}(B)$. Then, $\htwo(A|B)_{\rho|\sigma} \leqslant \hmin^{\eps}(A|B)_{\rho|\sigma} + \log \frac{2}{\eps^2} \leq \hmin^{\eps}(A|B)_{\rho} + \log \frac{2}{\eps^2}$.
\end{lem}
This lemma is very similar to Theorem 7 in \cite{tcr09}, but note that they use a slightly different definition of $\htwo$. The proof of this version of the lemma is, therefore, very similar to theirs.
\begin{proof}
	%First, note that $\hmin^{\eps}(A|B)_{\rho} \geqslant \hmin^{\eps}(A|B)_{\rho|\sigma}$. 
	Let $\Delta = (\rho_{AB} - 2^{-\hmin^{\eps}(A|B)_{\rho|\rho}} \ident_A \otimes \sigma_B)_{+}$ (where $(\cdot)_{+}$ denotes the nonnegative part of an operator). Let $\lambda > 0$ be such that $\hmin^{\eps}(A|B)_{\rho|\sigma} \geqslant -\log \lambda$ and $\eps = \sqrt{2 \tr[\Delta]}$ (such a $\lambda$ exists by Lemma 15 of \cite{tcr09}). Furthermore, let $P$ be the projector onto the support of $\Delta$. We then have that
	\begin{align*}
		P\rho_{AB} P &\geqslant \lambda P (\ident_A \otimes \sigma_B) P\\
		(P \sigma_B P)^{-1/2} P \rho_{AB} P (P \sigma_B P)^{-1/2} &\geqslant \lambda P_{AB},
	\end{align*}
	where we have omitted the $\ident_A$'s in the second line. Using this, we get that
	\begin{align*}
		\frac{\varepsilon^2}{2} &= \tr[\Delta]\\
		&= \tr[P(\rho_{AB} - \lambda \ident_A \otimes \sigma_B) P]\\
		&\leqslant \tr[P \rho_{AB} P]\\
		&\leqslant \lambda^{-1} \tr[P \rho_{AB} P (P \sigma_B P)^{-1/2} P\rho_{AB}P (P \sigma_B P)^{-1/2}]\\
		&= \lambda^{-1} 2^{D_2(P \rho_{AB} P \| P(\ident_A \otimes \sigma_B) P)}\\
		&\leqslant \lambda^{-1} 2^{D_2(\rho_{AB} \| \ident_A \otimes \sigma_B)}\\
		&\leqslant 2^{\hmin^{\varepsilon}(A|B)_{\rho|\sigma}-\htwo(A|B)_{\rho|\sigma}},
	\end{align*}
	where $D_2$ is defined in Definition \ref{def:d2} and the next to last inequality is due to Theorem \ref{thm:d2-monotone}.
\end{proof}

\begin{defin}\label{def:d2}
	Let $D_2(X \| Y)$ be defined as
	\[ 2^{D_2(X \| Y)} := \tr[(Y^{-1/4} X Y^{-1/4})^2]. \]
\end{defin}

\begin{thm}\label{thm:d2-monotone}
	$D_2( \mathcal{E}(X) \| \mathcal{E}(Y)) \leqslant D_2(X \| Y)$ for any CPTP map $\mathcal{E}$.
\end{thm}
\begin{proof}
	Consider the map $(L,R) \mapsto LR^{-1/2}L \otimes R^{-1/2}$. Theorem 5.14 in \cite{wolfnotes} shows that it is jointly operator convex, by taking $g(R) = R^{1/2} \otimes (R^{1/2})^{\top}$ (which is operator concave by \cite[Corollary 5.5, point 1]{wolfnotes}), $h(L) = L \otimes \ident$, $f(x) = x^2$. This entails that $(L,R) \mapsto \tr[R^{-1/2}LR^{-1/2}L]$ is also jointly operator convex, via the fact that
	\[ \tr[R^{-1/2}LR^{-1/2}L] = \tr[\Phi (LR^{-1/2}L \otimes (R^{-1/2})^{\top})]. \]
	We now invoke Theorem 5.16 from \cite{wolfnotes} on this functional to conclude the proof.
\end{proof}

Defining $\htwo(A|B)_{\rho|\sigma}$ naturally as $\htwo(A|B)_{\rho|\sigma} = \tr\left[ \left(\sigma_B^{-1/4} \rho_{AB} \sigma_B^{-1/4} \right)^2 \right]$, we obtain the following corollary.

\begin{cor}
	Let $\rho_{AB} \in \cS_{\leqslant}(AB)$ and $\sigma_B \in \cS_{\leqslant}(B)$ such that $\rho_B$ is in the support of $\sigma_B$. Then, for any CPTP map $\mathcal{E}_{B \rightarrow C}$, we have that $\htwo(A|B)_{\rho|\sigma} \leqslant \htwo(A|C)_{\mathcal{E}(\rho)|\mathcal{E}(\sigma)}$. 
\end{cor}

%WARNING: Do not uncomment unless you figure out what this was supposed to be
%\begin{lem}[{\cite[Lemma 18]{TSSR10}}]
%\label{lem:hmin-rhorho}
%Let $\eps'\geq0$, $\eps'>0$, and $\rho_{AB} \in \cS(AB)$. Then
%\begin{align*}
%\entHmin^{\eps}(A|B)_{\rho} - \log\left(\frac{2}{\eps'^2} + \frac{1}{1-\eps}\right) \leqslant \entHmin^{\eps+\eps'}(A|B)_{\rho} \leqslant \entHmin^{\eps+\eps'}(A|B)_{\rho}\ .
%\end{align*}
%\end{lem}

\begin{lem}
\label{lem:condition-cl}
Suppose $\rho \in \cS(AQC)$ is such that the $C$ system is classical, i.e.,
$\rho_{AQC} = \sum_{c} p(c) \proj{c} \otimes \rho^c_{AQ}$ for some probability distribution $p$ and orthogonal vectors $\{\ket{c}\}_c$ in $C$.
Then
\[
\entHtwo(A|QC)_{\rho} = - \log \sum_{c} p(c) 2^{-\entHtwo(A|Q)_{\rho^c}}.
\]
In particular $\entHtwo(A|QC) \geqslant -\log |Q|$.
\end{lem}
\begin{proof}
We have
\begin{align*}
&\tr \left[ \left(\id_{A} \otimes \rho_{QC}^{-1/4} \rho_{AQC} \id_{A} \otimes \rho_{QC}^{-1/4} \right)^2\right] \\ 
&=  \tr \left[ \left(\id_{A} \otimes (\sum_{c} p(c) \proj{c} \otimes \rho^c_{Q})^{-1/4} \left(\sum_{c} p(c) \proj{c} \otimes \rho^c_{AQ} \right) \id_{A} \otimes (\sum_{c} p(c) \proj{c} \otimes \rho^c_{Q})^{-1/4} \right)^2\right] \\
&= \sum_{c} \tr \left[ \left(\id_{A} \otimes ( p(c)\rho^c_{Q})^{-1/4} \rho^c_{AQ} \id_{A} \otimes (p(c)\rho^c_{Q})^{-1/4} \right)^2\right] \\
&= \sum_{c} p(c) \tr \left[ \left(\id_{A} \otimes (\rho^c_{Q})^{-1/4} \rho^c_{AQ} \id_{A} \otimes (\rho^c_{Q})^{-1/4} \right)^2\right].
\end{align*}
To conclude the proof, we simply observe that $\tr \left[ \left(\id_{A} \otimes (\rho_{Q})^{-1/4} \rho_{AQ} \id_{A} \otimes (\rho_{Q})^{-1/4} \right)^2\right] \leqslant \tr[\id_A \otimes \rho_Q^{-1} \rho_{AQ}] = \tr[ \rho_Q^{-1} \rho_Q ] = |Q|$.
\end{proof}
%\begin{claim}
%\label{claim:estimate}
%Write $h_2 = -1 + x$ with $x \leqslant 1/3$, then we have
%\[
%R_2(-1+x) \geqslant \min\left( 1, -\log\left(2-2 \frac{x}{10\log(1/x)}\right) \right).
%\]
%%We can also have
%%\[
%%R_2(-1+x) \geqslant \min\left( 1, -\log(2-\frac{x}{\log(2/x)}) \right).
%%\]
%\end{claim}
%\begin{proof}
%If $\frac{x}{10\log(1/x)} \in [0, 3/4]$, we can compute
%\begin{align*}
%f_2\left(\frac{x}{10\log(1/x)}\right) &= h\left(\frac{x}{10\log(1/x)}\right) + \frac{x}{10\log(1/x)} \log 3 - 1 \\
%				&\leqslant 2 \cdot \frac{x}{10\log(1/x)} \log \frac{10\log(1/x)}{x} + \frac{x}{10 \log(1/x)}\log 3 - 1\\
%				&\leqslant \frac{x}{5\log(1/x)} \left( \log 10 + \log \log(1/x) + \log(1/x) + \frac{\log 3}{2} \right) - 1 \\
%				&\leqslant x - 1,
%\end{align*}
%where we used the inequalities $\log \log (1/x) + \log(1/x) \leqslant 2 \log(1/x)$ and $\log 10 + \frac{\log 3}{2} \leqslant 3 \log(3) \leqslant 3 \log (1/x)$.
%\end{proof}

The following claim gives a bound on the function $\gamma$ from Theorem \ref{thm:ur-h2-bb84} for small values of $h_2$.
\begin{claim}
\label{claim:estimate-gamma}
Write $h_2 = -1 + x$ with $x \leqslant 1/3$, then we have
\[
\gamma(-1+x) \geqslant \frac{x}{10 \log(1/x)}. %\min\left( 1, -\log\left(2-2 \frac{x}{10\log(1/x)}\right) \right).
\]
%We can also have
%\[
%R_2(-1+x) \geqslant \min\left( 1, -\log(2-\frac{x}{\log(2/x)}) \right).
%\]
\end{claim}
\begin{proof}
Recall that $\gamma$ is the inverse of $g(x) = h(x) + x - 1$. We have
\begin{align*}
g\left(\frac{x}{10\log(1/x)}\right) &= h\left(\frac{x}{10\log(1/x)}\right) + \frac{x}{10\log(1/x)} - 1 \\
				&\leqslant 2 \cdot \frac{x}{10\log(1/x)} \log \frac{10\log(1/x)}{x} + \frac{x}{10 \log(1/x)} - 1\\
				&\leqslant \frac{x}{5\log(1/x)} \left( \log 10 + \log \log(1/x) + \log(1/x) + \frac{1}{2} \right) - 1 \\
				&\leqslant x - 1,
\end{align*}
which proves the desired result.
\end{proof}

\begin{lem}
\label{lem:sum-binomial}
Let $a$ be a positive integer. We have for any $\frac{\ell}{n} \leqslant \frac{a}{a+1}$,
\[
\sum_{k=0}^{\ell} \binom{n}{k} a^k \leqslant 2^{nh(\ell/n)} a^{\ell}
\]
\end{lem}
\begin{proof}
	See for example \cite{venkat-codingtheory}, Lemma 5.
\end{proof}

\section{Operational interpretation of $\htwo$}
\label{sec:oper-inter-h2}

When $X$ is classical, then it is already known~\cite{BCHLW08} that 
\begin{align*}
\entHtwo(X|E) = - \log P_{\rm guess}^{\rm pg}(X|E)\ , 
\end{align*}
where $P_{\rm guess}^{\rm pg}$ is the guessing probability using the pretty good measurement 
which performs very well~\cite{Hausladen94}.
For completeness, we here include the arguments of~\cite{bcw13} regarding the operational interpretation of $\htwo$ for quantum information $A$. 
Like the min-entropy, it is a natural measure of the entanglement between $A$ and $B$ in that $\htwo(A|B)=-\log [|A| F^{\rm pg}(A|B)^2]$ with
\begin{align}\label{eq:h2alt}
F^{\rm pg}(A|B)=F(\Phi^N_{AA'},\id_A \otimes \Lambda_{B\rightarrow A'}^{\rm pg}(\rho_{AB}))\ ,
\end{align}
and $\Lambda^{\rm pg}_{B\rightarrow A'}$ is the pretty good recovery map~\cite{BarnKnil02}. To see this, we note that the pretty good recovery map can be written as
\begin{align*}
\Lambda^{\rm pg}_{B\rightarrow A'}(\cdot)=\frac{1}{|A|}\cdot\mathcal{E}^\dagger_{B\rightarrow A'}\left(\rho_B^{-1/2}(\cdot) \rho_B^{-1/2}\right)\ ,
\end{align*}
where $\mathcal{E}^\dagger_{B\rightarrow A'}$ denotes the adjoint of the Choi-Jamiolkowski map of $\rho_{AB}$,
\begin{align*}
\mathcal{E}_{A\rightarrow B}(\cdot)=|A|\cdot\tr_{A}\left[\left((\cdot)^T \otimes \id_{B}\right)\rho_{AB}\right]\ .
\end{align*}
Putting this in~\eqref{eqn:def-h2} we arrive at~\eqref{eq:h2alt}. The map $\Lambda^{\rm pg}_{B\rightarrow A'}$ is pretty good in the sense that it is close to optimal for recovering the maximally entangled state, i.e., the following bound holds~\cite{BarnKnil02}
\begin{align*}
F^{2}(A|B)\leqslant F^{\rm pg}(A|B)\leqslant  F(A|B)\ ,
\end{align*}
where $F(A|B)$ is the fidelity achievable by the optimal map given in~\eqref{eq:minDef}.

\end{document}